\documentclass[acmsmall,nonacm]{acmart}

\usepackage[T1]{fontenc}
\usepackage{microtype}

\usepackage{cleveref}
\usepackage{enumitem}
\usepackage{thmtools}
\usepackage{mathtools}

\usepackage{tikz}
\usetikzlibrary{arrows.meta,positioning}

\theoremstyle{acmdefinition}

\crefname{lemma}{lemma}{lemmas}
\Crefname{lemma}{Lemma}{Lemmas}
\crefname{definition}{definition}{definitions}
\Crefname{definition}{Definition}{Definitions}

\Crefname{corollary}{Corollary}{Corollaries}

\Crefname{proposition}{Proposition}{Propositions}
\crefname{conjecture}{conjecture}{conjectures}
\Crefname{conjecture}{Conjecture}{Conjectures}

\Crefname{theorem}{Theorem}{Theorems}

\mathchardef\mhyphen="2D
\newcommand{\ghw}{\mathsf{ghw}}
\newcommand{\adw}{\mathsf{adw}}
\newcommand{\subw}{\mathsf{subw}}

\newcommand{\fhw}{\mathsf{fhw}}

\newcommand{\adwlift}{\adw^{\mathrm{lift}}}

\newcommand{\subwlift}{\subw^{\mathrm{lift}}}

\newcommand{\tw}{\mathsf{tw}}

\newcommand{\abs}[1]{\left|#1\right|}
\newcommand{\1}{\mathbf{1}}

\newcommand{\Mod}{\mathsf{Mod}}
\newcommand{\pMod}[1]{#1\mhyphen\Mod}

\newcommand{\epr}{\mathbf{p}}
\newcommand{\Sub}{\mathsf{Sub}}
\newcommand{\pSub}[1]{#1\mhyphen\Sub}
\newcommand{\rpos}{\mathbb{R}_{\ge 0}}

\newcommand{\Inc}{\operatorname{Inc}}
\newcommand{\disp}{\mathfrak{B}}
\newcommand{\gauge}{\gamma}

\newcommand{\Sym}{\mathcal S}
\newcommand{\one}{\mathbf 1}

\newcommand{\cut}[1]{\mathsf{cut}_#1}

\newcommand{\eex}{\mathsf{ex}}
\newcommand{\cgamp}{\mathfrak{R}}

\newcommand{\cqprob}{\mathsf{BCQ}}
\newcommand{\pcq}{\mathsf{p}\mhyphen{}\mathsf{BCQ}}

\newcommand{\cH}{\mathcal{H}}

\title{Cuts and Gauges for Submodular Width}

\author{Matthias Lanzinger}
\orcid{0000-0002-7601-3727}
\authornote{Supported by the Vienna Science and Technology Fund (WWTF) [10.47379/ICT2201].}
\affiliation{
  \institution{Institute of Logic and Computation, TU Wien}
  \city{Vienna}
  \country{Austria}
}
\email{matthias.lanzinger@tuwien.ac.at}

\begin{document}

\begin{abstract}
Submodular width is a central structural measure governing the complexity of
conjunctive query evaluation. In this paper we recast submodular width in
geometric terms. We show that submodular width can be approximated, up to a factor
$3/2$, by a new branchwidth parameter defined in terms of edge separations in the
hypergraph and the costs induced on them by admissible submodular
functions. 
This reformulation turns lower bounds on submodular width into the problem of
constructing well-balanced edge separations whose induced cost remains small.
We then express this connection through a variational characterisation in terms of a convex body.
Using these tools, we relate submodular width to more familiar graph-theoretic
notions, including line-graph treewidth and multicommodity flow, and obtain
general conditions under which submodular width is tightly linked to
generalised hypertree width. In particular, under various natural conditions we show that
\[
\subw(H) \in \Omega \left(\frac{\ghw(H)}{\log\ghw(H)} \right).
\]
\end{abstract}

\maketitle

\tableofcontents

\clearpage

\section{Introduction}
\label{sec:intro}
\paragraph{Context} Evaluation of conjunctive queries (CQs) refers to the problem of, given a relational FO formula $q$ built using only $\exists$ and conjunction and a relational structure $D$, deciding whether $q$ is satisfied in $D$. This represents one of the core algorithmic questions underlying data management but is also just as fundamental for a broad range of fields in computer science: it is equivalent to deciding the existence of a homomorphism between relational structures and to deciding constraint satisfaction problems. The general problem is $\mathsf{NP}$-complete, but its wide use has inspired highly detailed research into the specific factors that make it hard or easy. For instance, recently Bulatov~\cite{Bulatov17} and Zhuk~\cite{Zhuk20} independently proved a $\mathsf{PTIME}$/$\mathsf{NP}$ dichotomy for the problem depending on the form of $D$, concluding a decades long major research programme of understanding that side of the problem.

On the other hand, the question of how the form of $q$ affects the complexity of CQ evaluation has also received major attention. Especially in a database setting, the structure of $q$ corresponds to the structure of a database query, and positive algorithmic results directly transfer into more efficient algorithms for database query evaluation.
In this direction, the ``cyclicity'' of $q$ was identified as the major source of hardness. For constantly bounded arity of $q$, Grohe~\cite{DBLP:journals/jacm/Grohe07} showed that CQ evaluation is in $\mathsf{PTIME}$ exactly for queries of constantly bounded treewidth. When arity is not considered constant, generalisations of hypergraph $\alpha$-acyclicity --- such as generalised and fractional hypertree width ($\ghw$ and $\fhw$, respectively)~\cite{DBLP:journals/jacm/GottlobMS09,fhw} ---  were shown sufficient for $\mathsf{PTIME}$ algorithms, even when treewidth is unbounded. However, no matching lower bound is known in the unbounded rank setting (except for the special case of maximum hypergraph degree 2~\cite{LanzingerPODS22}).

From the perspective of parameterised complexity (one considers the query $q$ as the parameter), the situation is clearer. Seminal work of Marx~\cite{Marx13} introduced \emph{submodular width} ($\subw$), and showed that bounded $\subw$ characterises exactly those classes of queries for which evaluation is fixed-parameter tractable. This in turn has motivated the development of important new algorithmic techniques that achieve $f(q)\mathsf{poly}(D)^{\subw(q)}$ time~\cite{panda25,jaguar,pandaexpress}.

Despite its algorithmic importance, the structural aspects of $\subw$ remain opaque. Bounded $\subw$ is the exact boundary for fixed-parameter tractability, while bounded $\ghw$ and $\fhw$ are the standard sufficient conditions for polynomial-time tractability, yet the relationship between these parameters is largely unresolved. While we know that $\subw(H) \le \fhw(H) \le \ghw(H)$, neither inequality can be reversed generally\footnote{\citet{DBLP:journals/jacm/GottlobLPR21} show reverse inequalities of the form $\ghw(H) \le f(\fhw(H))$ under a variety of conditions.}. However, unbounded $\fhw$ (and hence $\ghw$) with bounded $\subw$ is currently only known to occur in pathological examples~\cite{adw}. In classical structural settings, they are only known to coincide under bounded rank and maximum degree 2~\cite{LanzingerPODS22}, but even there, the known relationship is purely qualitative, offering no insight into the quantitative relationship of the two parameters. 

The root of this disconnect is a lack of appropriate mathematical tools. Establishing lower bounds on $\subw$ requires identifying an edge-dominated submodular witness $b$ such that \emph{every} tree decomposition of $H$ contains at least one bag with a large $b$-weight. The community has struggled to find general techniques to systematically construct these witnesses. Consequently, fundamental gaps remain in our understanding: we lack reliable ways to computationally check or certify submodular width, its exact relationship to measures like adaptive width is unmapped, and we do not even know if $\subw$ yields the optimal exponent for evaluation time. Bridging the gap between polynomial and fixed-parameter tractability for CQ evaluation demands a completely new analytical framework for $\subw$.

\paragraph{Contribution}
Our main contribution is a reformulation of submodular width  in terms of symmetric cut functions and finite-dimensional convex geometry. The starting point is a branchwidth-type parameter, \emph{lifted submodular width}, obtained by applying admissible submodular witnesses to edge boundaries. We prove that for every hypergraph $H$,
\[
\subwlift(H)\le \subw(H)\le \max\Bigl\{1,\frac32\,\subwlift(H)\Bigr\}.
\]
That is, $\subwlift$ is equivalent to $\subw$ up to constant factor $3/2$. This replaces tree decompositions by branch decompositions over certain cut cost functions.

While this branchwidth perspective simplifies the structural decomposition, it still restricts us to cut weightings explicitly induced by submodular witnesses. Our second step removes this restriction by introducing a variational description of $\subwlift$ framed in terms of finite-dimensional convex geometry. We define the submodular realisable body $K_\one^\Sub(H)$ and its gauge\footnote{Intuitively, the gauge of a point is the minimum factor by which a fixed reference set must be scaled to contain it, serving as a geometric measure of its complexity.} $\gauge_H^\Sub$, demonstrating that
\[
\subwlift(H)=
\sup_{q\in \Sym(H)}
\frac{\disp_H(q)}{\gauge_H^\Sub(q)}
\]
where $\Sym(H)$ is the set of all symmetric edge set functions $2^E \to \rpos$.
This variational characterisation provides us with a significantly more technical freedom than previous formulations of submodular width. Instead of constructing cut weightings that correspond to a submodular function directly, the geometry of $K_\one^\Sub(H)$ allows us to use the much broader space of arbitrary symmetric edge-set functions to construct any $q$ that exhibits a large branchwidth value relative to its gauge.
These first two results taken together provide us with an entirely new approach for reasoning about the submodular width of hypergraphs.

Famously, freedom dies unless it is used. The second half of our contribution thus presents one way to build on this reframing of submodular width. Namely, we present a method for systematically constructing a kind of canonical function in $\Sym(H)$. Drawing on classic results connecting well-linked sets and capacitated multicommodity flows in graphs, we prove that large line graph treewidth guarantees the existence of a $q \in \Sym(H)$ with $\disp_H(q)\ge 4/9$ while bounding the local incident load on any single hyperedge $e \in E(H)$ to $O(\log \tw(M) / \tw(M))$, where $M$ is the line graph of $H$.

We subsequently formalise the translation of this bounded incident load into an upper bound on the gauge $\gamma_H^\Sub(q)$. To this end we introduce the \emph{submodular gauge-routing ratio} $\cgamp(H)$, as a way measures how efficiently low incident load from the previous step can be transferred into an upper bound on $\gamma_H^\Sub(q)$. Concretely, we obtain the general lower bound
\[
\subw(H) \in \Omega\left(\frac{1}{\cgamp(H)} \cdot \frac{\ghw(H)}{\log \ghw(H)}\right).
\]
We then also relate this to concrete structural properties of the hypergraph. For instance, if $H$ satisfies the \emph{private intersection property} (each pairwise intersection contains a degree-2 vertex), then $\cgamp(H) \le 1$. Alternatively, if $H$ has bounded edge excess $\eex(H)$ (the sum of degrees in an edge, ignoring degree 2 vertices), then $\cgamp(H) \le \eex(H) + 1$. That is, we significantly generalise the previous best known conditions under which $\ghw$ and $\subw$ are tightly connected, and moreover, we give explicit quantitative bounds that were missing for previous work on the degree 2 case. We note that these examples are specific, easy to present, instances of our results. We additionally identify a more general principle under which $\cgamp(H)$ is small.
These structural results additionally demonstrate that for a wide range of structures, polynomial time tractability and fixed-parameter tractability coincide for CQ evaluation. 
\begin{theorem}[Informal version of \Cref{thm:complexity}]
  Let $\cH$ be a recursively enumerable class of hypergraphs with bounded submodular gauge-routing ratio, i.e.,  $\cgamp(\cH)<\infty$. 
Assuming ETH, the following are equivalent:
\begin{enumerate}[label=(\roman*)]
\item $\ghw(\cH) < \infty$.
\item $\subw(\cH)< \infty$.

\item CQ evaluation over queries with hypergraphs in $\cH$ is tractable.
\item  CQ evaluation, parameterised by query size, over queries with hypergraphs in $\cH$ is fixed-parameter tractable.
\end{enumerate}
\end{theorem}

\paragraph{Organization}
The remainder of the paper is structured as follows. 
Necessary technical preliminaries are introduced in~\Cref{sec:prelims}.
\Cref{sec:lift} formalises the boundary lift, establishing the core equivalence between submodular width and lifted submodular width.
\Cref{sec:geometry} develops the continuous geometry of the realisable submodular body and its dual antiblocker.
\Cref{sec:support-inputs} shows how to construct the universal balanced cuts via node-capacitated multicommodity flows.
\Cref{sec:applications} combines the results of the previous three sections to show new lower bounds of $\subw$ in terms of $\ghw$ under a variety of conditions.
Finally, \Cref{sec:consequences} formalizes the complexity theoretic consequences for CQ evaluation, and \Cref{sec:conclusion} outlines the rich avenues for future research exposed by our new geometric approach to submodular width.

\section{Preliminaries}
\label{sec:prelims}
For a set of sets $X$ we sometimes write $\bigcup X$ as shorthand for $\bigcup_{s\in X} s$. By convention, if $X=\emptyset$, then $\bigcup X = \emptyset$. For set $X$ we also write $\binom{X}{k}$ for the set of all subsets of $X$ with cardinality $k$.

\paragraph{Graphs \& Hypergraphs}
A \emph{hypergraph} $H$ is a pair $(V,E)$ where $V$ is a set of objects and $E \subseteq 2^V$. The \emph{degree} of a vertex $v\in V$, also written $\deg_H(v)$ is the number of edges that contain $v$. The degree of $H$ is $\deg(H) := \max_{v\in V} \deg_H(v)$. Throughout, we assume that hypergraphs have no isolated vertices, i.e., no vertices with degree $0$.

When there are multiple (hyper)graphs under discussion, we use the notation $V(H)$ and $E(H)$ to refer specifically to the vertex and edge sets of $H$.

The \emph{primal graph} (also called Gaifman graph) of a hypergraph $H$ is the graph $P(H)$ on vertex set $V(H)$ in which two distinct vertices $v,w$ are adjacent if $\exists e \in E(H)$ such that $v,w\in e$.
The \emph{dual hypergraph} $H^*$ of $H$ has $E(H)$ as its vertices, and for every $v \in V(H)$ it has the hyperedge $f_v = \{e \in E(H) \mid v \in e\}$.
The graph $P(H^*)$ will be referred to as the line graph of $H$.

For a hypergraph $H=(V,E)$ and a vertex $x\in V$, write
$\Inc_H(x):=\{e\in E \mid x\in e\}$
for the set of edges incident with $x$.
For $\alpha:E\to \rpos$ and $F\subseteq E$, write
$\alpha(F):=\sum_{e\in F}\alpha_e$.
Thus $\alpha(\Inc_G(x))$ is the total $\alpha$-weight of edges incident with $x$.
When $G=P(H^*)$ is fixed, we abbreviate to
$\Inc(x):=\Inc_{P(H^*)}(x)$.

\paragraph{(Sub)Modular Set Functions}
We will be interested in functions of the form
\[
b:2^{S}\to \mathbb R_{\ge 0}.
\]
that map sets to non-negative reals.
We say that such a function is \emph{normalized} if $b(\emptyset)=0$, \emph{monotone} if $b(X)\le b(Y)$ whenever $X\subseteq Y$.

It is \emph{modular} if there are weights $(w_v)_{v\in S}$ with $w_v\ge 0$ such that
\[
b(X)=\sum_{v\in X} w_v
\qquad \forall X\subseteq S.
\]
It is \emph{submodular} if
$b(X)+b(Y)\ge b(X\cup Y)+b(X\cap Y)$ for all $X,Y\subseteq S$.

Let $\Mod(H)$ be the set of all functions $2^{V(H)} \to \mathbb{R}_{\ge 0}$ that are normalised, and modular (and hence implicitly monotone). 
An \emph{edge profile} $\epr$ for hypergraph $H=(V,E)$ is a function $E(H) \to \mathbb{R}_{\ge 0}$. Moreover, for an edge profile $\epr$ define 
$$\pMod{\epr}(H) := \{ b \in \Mod(H) \mid b(X)\le \epr(e) \quad \forall e \in E, \forall X \subseteq e\}.$$
Similarly, $\Sub(H)$ is the set of normalised, monotone, submodular functions, with $\pSub{\epr}(H)$ defined analogously to $\pMod{\epr}$.
We write $\one$ for the edge profile that maps every edge to $1$. $\pMod{\one}$ then corresponds to the so-called \emph{edge-dominated} normalised, monotone, and modular functions.

\paragraph{Branch decompositions}
For hypergraph $H=(V,E)$, a \emph{branch decomposition} of $E$ is a pair $(T,\delta)$ where $T$ is a subcubic (i.e., every vertex has at most degree 3) tree and $\delta$ is a bijection from the leaves of $T$ onto $E$. Every edge $f\in E(T)$ induces a bipartition
\(
(X_f, E\setminus X_f)
\)
of the hyperedge set according to the two components of $T-f$.

We call a function $q:2^{E}\to \mathbb R_{\ge 0}$ an \emph{edge set function} (of $H$). We say such a function is symmetric if $q(X)=q(E\setminus X)$. We will often refer to normalised symmetric edge set functions as \emph{cut weightings}.
For cut weighting $q$ define
\[
\disp_H(q):=\min_{(T,\delta)}\max_{e\in E(T)} q(X_e),
\]
where the minimum runs over branch decompositions $(T,\delta)$ of $E(H)$. By convention, if $|E|\le 1$, $\disp_H(q)$ is always $0$. 

That is, $\disp_H$ is simply branchwidth over some symmetric normalised weighting function for the cuts, similar to previous work on branchwidth of connectivity functions, see e.g., \cite{OumSeymour06,Grohe16}. Note however that we generally do not deal with connectivity functions here as our cut weightings are not necessarily submodular.

\paragraph{Tree decompositions}
A tree decomposition of a hypergraph $H$ is a pair $(T,B)$, where $T$ is a tree and $B : V(T)\to 2^{V(H)}$ such that:
\begin{itemize}
\item for each $e \in E(H)$ there is a node $u \in V(T)$ such that $e \subseteq B(u)$.
\item for every vertex $v \in V(H)$, the set of nodes $\{u \in V(T) \mid v \in B(u)\}$ is non-empty and forms a connected subgraph of $T$.
\end{itemize}
For $b : 2^{V(H)} \to \mathbb{R}$, the \emph{$b$-width} of a tree decomposition $(T,B)$ is  $\max_{u \in V(T)} b(B(u))$.
The \emph{$b$-width} of $H$ is defined as
\[
\tau_b(H) := \min_{(T,B)} \max_{u\in V(T)} b(B(u))
\]
where the minimum ranges over all tree decompositions $(T,B)$ of $H$. The \emph{treewidth} $\tw(H)$ of $H$ is $\tau_b$ where $b : U \mapsto |U|-1$. The \emph{generalised hypertree width} $\ghw(H)=\tau_b$ where $b$ maps set $U \subseteq V(H)$ to its edge cover number~\cite{DBLP:journals/jacm/GottlobMS09}.

The following is well known and easy to verify (see, e.g.,~\cite{LanzingerPODS22}).
\begin{proposition}\label{prop:support-tw-controls-ghw}
For every hypergraph $H$ it holds that $\ghw(H)\le \tw(H^*)+1$.
\end{proposition}

The \emph{adaptive width} $\adw$~\cite{adw} and \emph{submodular width} $\subw$~\cite{Marx13} of $H$ are defined as  
\[
\adw(H) = \sup_{b \in \pMod{\one}(H)} \tau_b(H) \qquad \subw(H) = \sup_{b \in \pSub{\one}(H)} \tau_b(H)
\]

It is known that $\ghw$ is always at least as high as $\subw$. Much of this paper is concerned with the question under which conditions the two measures are linked more closely.
\begin{proposition}[\citet{Marx13}]
\label{ghw.subw}
For every hypergraph $H$ it holds that $\adw(H) \le \subw(H) \le \ghw(H)$.
\end{proposition}

\paragraph{Cones \& Gauges}

Let $X$ be a real vector space. A set $C\subseteq X$ is a \emph{cone} if
whenever $x\in C$ and $\alpha\ge 0$, we also have $\alpha x\in C$. In other
words, together with every point, the set contains the whole ray starting at
the origin and passing through that point.

A cone $C$ is \emph{convex} if whenever $x,y\in C$ and $\alpha,\beta\ge 0$, we
have $\alpha x+\beta y\in C$. Equivalently, $C$ is closed under addition and
under multiplication by nonnegative scalars.

If $X$ is equipped with a topology (for example, if $X=\mathbb{R}^n$ with the
usual topology), then a convex cone $C\subseteq X$ is a \emph{closed convex
cone} if it is closed as a subset of $X$.

For any nonempty set $A\subseteq \mathbb R_{\ge 0}^{N}$, define its \emph{gauge} (also called a Minkowski functional) as
\[
\gauge_A(x):=\inf\{\alpha\ge 0\mid x\in \alpha \cdot A\},
\]
with value $+\infty$ if no such $\alpha$ exists.

A multivariate function $f(x_1,\dots,x_n)$ is \emph{degree-1 positively homogeneous} if
\[
f(\alpha \, x_1,\dots, \alpha\, x_n) = \alpha\, f(x_1,\dots,x_n)
\]

\section{Lifted Cuts and Widths}
\label{sec:lift}

In this section we introduce our first major reframing of submodular width. Part of the technical difficulty in understanding the submodular width of a hypergraph lies in reasoning about $\tau_b$ for arbitrary $b\in\pSub{\one}$. In this section we identify an alternative way to understand the problem by instead considering  cuts of the hypergraph, that themselves are weighted by some symmetric edge set function $\lambda : 2^E \to \rpos$. 

Concretely, we introduce a branchwidth analogue of submodular width (and adaptive width) that is equivalent to $\subw$ up to a constant factor $3/2$. The key insight in constructing this branchwidth parameter is to base it specifically on cut weightings that are induced by $b \in \pSub{\one}$.

\begin{definition}[Boundary Lift]
\label{def:boundarylift}
For hypergraph $H=(V,E)$, define the \emph{boundary function} 
\[
\partial_H(F):= \left(\bigcup F\right)\cap\left(\bigcup (E\setminus F)\right)
\]
consisting of all vertices that are in both sides of the edge bipartition induced by $F$.

For normalised and monotone $b: 2^V \to \rpos$ define the \emph{boundary lift} of $b$ as $\lambda_b := b \circ \partial_H$, i.e., 
\[
\lambda_b(F)=b(\partial_H(F)) \qquad \forall F \subseteq E
\]
\end{definition}

\begin{proposition}
\label{prop:lift.sym}
Let $b : {2^V} \to \rpos$ be normalised. Then its boundary lift $\lambda_b$  is a normalised symmetric edge set function.
\end{proposition}
\begin{proof}
Recall that $b$ is normalised and  $\partial_H(\emptyset)=\emptyset$, hence $\lambda_b(\emptyset)=b(\emptyset) =0$. 
Symmetry is immediate from $\partial_H(F)=\partial_H(E\setminus F)$.
\end{proof}

\begin{definition}[Lifted width]\label{def:lifted-width}
The \emph{lifted submodular width} of $H$ is
\[
\subwlift(H):=
\sup\{\disp_H(\lambda_b)\mid b\in \pSub{\one}(H) \}.
\]
Analogously, define \emph{lifted adaptive width} $\adwlift(H) = \sup\{\disp_H(\lambda_b)\mid b\in \pMod{\one}(H) \}$.
\end{definition}

Lifted submodular width is thus a measure of how expensive it is to cut a hypergraph, under certain cost functions. Importantly, this is in fact very closely tied to standard submodular width. This is best understood by thinking of \(\partial_H(F)\) as the ``boundary'' of an edge separation. Whenever the hyperedges are split into \(F\) and \(E(H)\setminus F\), the vertices in \(\partial_H(F)\) are exactly those that still connect the two sides. Any decomposition that separates the two parts must therefore keep these vertices visible at the separator. Branch decompositions measure the size of this boundary directly, whereas tree decompositions measure the size of bags needed to contain it. So the two widths are measuring the same phenomenon from two different points of view.

\begin{theorem}
\label{thm:lift.to.tau}
Let $H=(V,E)$ be a hypergraph and let $b\in \pSub{\one}(H)$.
Then
\[
\disp_H(\lambda_{b})\le \tau_{b}(H)\le
\max\Bigl\{1,\frac32\,\disp_H(\lambda_{b})\Bigr\}.
\]
\end{theorem}

\begin{proof}
For the first inequality, let $(T,B)$ be a tree decomposition of $H$ of
$b$-width $w$.
For each hyperedge $e\in E(H)$ choose a node $u_e\in V(T)$ with $e\subseteq B(u_e)$, and attach a new leaf $\ell_e$ to $u_e$ labelled by $e$. Let $T^+$ be a minimal subtree that spans all leaves $\ell_e$ of the resulting tree. Replace every vertex of $T^+$ of degree larger than $3$ by an arbitrary subcubic refinement on the same incident branches. Every internal node created from the subcubic refinement of a node $t$ is also assigned the bag $B(t)$. Thus, the resulting leaf-labelled tree is a branch decomposition $(T',\delta)$ of $E(H)$.

We claim that every cut of $(T',\delta)$ is separated by some original bag $B(u)$. Indeed, if $f$ is an edge of $T'$ incident to a leaf $\ell_e$, then
\[
\partial_H(X_f)\subseteq e\subseteq B(u_e).
\]
If $f$ comes from an original edge $u\,u'$ of $T$, then any vertex meeting hyperedges on both sides of the cut must belong to both $B(u)$ and $B(u')$, because the bags containing that vertex form a connected subtree crossing the edge $u\,u'$. Finally, if $f$ was created inside the local subcubic refinement of a node $u$, then the two sides of the cut correspond to a partition of the branches incident with $u$. Any vertex meeting hyperedges on both sides again has its bag-subtree meeting two different incident branches at $u$, and therefore contains $u$ itself. In every case,
\(
\partial_H(X_f)\subseteq B(u)
\)
for some original node $u$.

Now fix such a cut $X_f$. Since $\partial_H(X_f)\subseteq B(u)$ and $b$ is monotone,
\[
\lambda_b(X_f)=b(\partial_H(X_f))\le b(B(u))\le w.
\]
Thus every cut of $(T',\delta)$ has $\lambda_b$-value at most $w$, and therefore
\(
\disp_H(\lambda_b)\le \tau_b(H).
\)

For the second inequality, let $(T,\delta)$ be a branch decomposition of $E$
with width $k=\disp_H(\lambda_{b})$.
Suppressing degree-$2$ vertices -- i.e., contracting the two incident edges of each such vertex -- does not change the cuts, so we may assume that every internal vertex of $T$ has degree $3$.
If $|E(H)|\le 1$, then by convention $\disp_H(\lambda_b)=0$, and the trivial
one-bag tree decomposition gives $\tau_b(H)\le 1$, so the claim is immediate.

Assume $|E(H)|\ge 2$. For each leaf $\ell$ labelled by a hyperedge $e$, define $B(\ell)=e$.

For each internal node $u$ of $T$, let $F_1,F_2,F_3$ be the three hyperedge sets corresponding to the components of $T-u$, and define
\[
B(u):=
\left(\bigcup F_1\cap  \bigcup F_2 \right)
\cup
\left(\bigcup F_2\cap \bigcup F_3\right)
\cup
\left(\bigcup F_3\cap \bigcup F_1\right).
\]
We claim that these bags form a tree decomposition of $H$. Every hyperedge $e$ is contained in its leaf bag $B(\ell)=e$. Now fix $v\in V(H)$, and let $T_v$ be the minimal subtree of $T$ spanning all leaves labelled by hyperedges containing $v$. Then a node $u$ satisfies $v\in B(u)$ if and only if $u\in V(T_v)$: this is immediate for leaves, and for an internal node $u$ with induced edge sets $F_1,F_2,F_3$, the vertex $v$ belongs to $B(u)$ exactly when hyperedges containing $v$ occur in at least two of the three components of $T-u$, which is equivalent to $u\in V(T_v)$. Hence the bags containing $v$ form the connected subtree $T_v$.

Fix an internal node $u$, and write
$S_i:=\partial_H(F_i)$  for $i=1,2,3$.
If $x\in B(u)$, then $x$ lies in at least two of the three sets
$X_i:=\bigcup F_i$ for $i=1,2,3$.

Suppose that $x\in X_i\cap X_j$ with $i\neq j$, then $x$ is incident with some hyperedge of $F_i$ and some hyperedge outside $F_i$, so $x\in \partial_H(F_i)=S_i$. Likewise $x\in S_j$. Therefore every $x\in B(u)$ belongs to at least two of $S_1,S_2,S_3$. In particular,
\[
B(u)\subseteq U:=S_1\cup S_2\cup S_3
\qquad\text{and}\qquad
B(u)\subseteq I:=(S_1\cap S_2)\cup(S_2\cap S_3)\cup(S_3\cap S_1).
\]

Applying submodularity thrice (always on the two left-most terms) yields a single chain of inequalities:
\begin{align*}
b(S_1) + b(S_2) + b(S_3) &\ge b(S_1 \cup S_2) + b(S_3) + b(S_1 \cap S_2) \\
 &\ge b(S_1 \cup S_2) 
+ b\big((S_1 \cap S_2) \cup S_3\big) 
+ b(S_1 \cap S_2 \cap S_3)\\
& \ge b(\underbrace{S_1 \cup S_2 \cup S_3}_{U}) 
  + b\big(\underbrace{(S_1 \cup S_2) \cap ((S_1 \cap S_2) \cup S_3)}_I\big)
 + \underbrace{b(S_1 \cap S_2 \cap S_3)}_{\ge 0}
\end{align*}
Observe that
\(
(S_1\cup S_2)\cap((S_1\cap S_2)\cup S_3)
=
(S_1\cap S_2)\cup(S_2\cap S_3)\cup(S_3\cap S_1)=I.
\)
Now, by nonnegativity of $b$, we may drop the last term and by monotonicity $b(U), b(I)\ge b(B(u))$, hence
\[
b(S_1) + b(S_2) + b(S_3) \ge b(U) + b(I) \ge 2\,b(B(u)).
\]
Using $\lambda_b(F_i)=b(S_i)$ and $\lambda_b(F_i)\le k$, we obtain
\[
2\,b(B(u)) \le \lambda_b(F_1)+\lambda_b(F_2)+\lambda_b(F_3) \le 3k.
\]

Thus $b(B(u))\le \frac32\,k$ for any internal node $u$.

 Leaf bags satisfy $b(B(\ell))=b(e)\le 1$ because $b\in\pSub{\one}(H)$, and  therefore also
\[
\tau_{b}(H)\le \max\Bigl\{1,\frac32\,\disp_H(\lambda_{b})\Bigr\}.
\]
\end{proof}

\begin{corollary}
\label{cor:lifted-bw}
For every hypergraph $H$,
\[
\subwlift(H)\le \subw(H)\le
\max\Bigl\{1,\frac32\,\subwlift(H)\Bigr\}.
\]
\[
\adwlift(H)\le \adw(H)\le
\max\Bigl\{1,\frac32\,\adwlift(H)\Bigr\}.
\]

\end{corollary}

At first sight it might not be clear what is gained by this step. We argue that~\Cref{cor:lifted-bw} in fact has real technical benefits. First, cuts are simpler to technically analyse than tree decompositions. Specifically to argue for high $\subw$, one must argue that for some $b\in \pSub{\one}$, some expensive bag under $b$ must occur in every tree decomposition. In general this is a highly challenging task, and it remains a major problem of current database theory how to do so. In the lifted setting the situation is in a sense simpler, we give weights to cuts $F, E\setminus F$, and roughly speaking it is then sufficient to demonstrate a cut weighting such that every bipartition of $E$ has large weight. To illustrate this point, we very simply prove a fairly general submodular width lower bound for a wide range of ``hypergrid'' style hypergraphs using boundary lifts in~\Cref{app:lift.ex}.

However, alone this approach is still limited. In particular, we are limited to considering cut weightings that are induced by a submodular function in the sense of \Cref{def:boundarylift}. However, we will show in the next section that this weakness cannot only be mitigated, but turned into a strength.

\section{The Geometry of the Submodular Realisable Body}
\label{sec:geometry}

Throughout this section assume a fixed hypergraph $H=(V,E)$. We identify each
set function $b:2^V\to \mathbb R$ with its coordinate vector
$(b(X))_{X\subseteq V}\in \mathbb R^{2^V}$, and we freely pass between the
function and vector viewpoints. Accordingly, for $\alpha\in \mathbb R$ and set
functions $b,c:2^V\to \mathbb R$, scalar multiplication and addition are taken
pointwise:
\[
(\alpha \cdot b)(X):=\alpha\cdot b(X),\qquad (b+c)(X):=b(X)+c(X)\qquad  \forall X\subseteq V.
\]
Likewise, inequalities between set functions are understood pointwise, i.e., $b \le c$ means $b(X) \le c(X)$ for all $X \subseteq V$.

We use standard facts from analysis, finite-dimensional convex geometry and topology. 
In particular, we freely use that linear maps are continuous, that continuous images and coordinate projections of compact sets are compact, and that closed subsets of compact sets are compact. 
For a comprehensive treatment of these topics, we refer the reader to Rockafellar~\cite{rockafellar1997convex}.

The following additional notation will be convenient.
For a set $A\subseteq \rpos^N$, write
\[
\downarrow \!A:=\{x\in \rpos^N \mid \exists a\in A \text{ with } x\le a\}
\]
for its (pointwise) downward closure.
Moreover, we let $\Sym(H)$ denote the set of all symmetric nonnegative edge-set functions of
$H$, that is,
\[
\Sym(H):=
\{q:2^E\to \rpos \mid q(F)=q(E\setminus F)\ \forall F\subseteq E\}.
\]

\begin{definition}[Boundary Lift Operator]
\label{def:liftingmap}
Define the \emph{boundary-lift operator} $\Lambda_H:\mathbb R^{2^V}\to \mathbb R^{2^E}$
as
\[
\Lambda_H :  b \mapsto b \circ \partial_H.
\]

\end{definition}

Even when $b$ is normalised and submodular on $2^V$, its boundary lift
$\Lambda_H b$ need not be submodular on $2^E$.
Thus $\Lambda_H b$ is not, in general, a connectivity function.
Nevertheless, symmetry and nonnegativity suffice for $\disp_H$.

\begin{lemma}\label{lem:sub-slice-geometry}
For every edge profile $\epr$, the set $\pSub{\epr}(H)$ is convex and compact.
\end{lemma}

\begin{proof}
The set $\pSub{\epr}(H)$ is convex because normalization, monotonicity,
submodularity, and the constraints $b(e)\le \epr(e)$ for $e\in E$ are all
preserved under convex combinations.

Moreover, $\pSub{\epr}(H)\subseteq \mathbb R^{2^V}$ is defined by finitely many
closed conditions:
\begin{align*}
& b(\emptyset)=0, \qquad  b(X)\le b(Y) \qquad \forall X \subseteq Y \subseteq V, \\
&b(X)+b(Y)\ge b(X\cup Y)+b(X\cap Y)\qquad\forall X,Y\subseteq V,
\end{align*}
and $b(e)\le \epr(e)$  for all $e\in E$.
Hence $\pSub{\epr}(H)$ is closed.

To show boundedness, choose for each vertex $v\in V$ a hyperedge $e_v\in E$
with $v\in e_v$. Recall that we assume throughout that every vertex of $H$ lies in some hyperedge. By monotonicity,
\[
0\le b(\{v\})\le b(e_v)\le \epr(e_v).
\]
Now recall that normalised submodular functions are subadditive, i.e.,
\[
b(X\cup Y)\le b(X)+b(Y)\qquad\forall X,Y\subseteq V.
\]
Repeated application of this inequality then also implies
\[
b(X)\le \sum_{v\in X} b(\{v\})\qquad\forall X\subseteq V.
\]
Therefore
\[
0\le b(X)\le \sum_{v\in X}\epr(e_v)\le \sum_{v\in V}\epr(e_v)
\qquad\forall X\subseteq V.
\]
So every coordinate $b(X)$ is uniformly bounded on $\pSub{\epr}(H)$.

Thus $\pSub{\epr}(H)$ is closed and bounded in a finite-dimensional space, and
therefore compact.
\end{proof}

\begin{definition}[Realisable Body and Submodular Gauge]
For an edge profile $\epr$, define the \emph{submodular realisable body}
\[
K^\Sub_{\epr}(H):=
\downarrow \!\Lambda_H(\pSub{\epr}(H)).
\]

In other words, $K^\Sub_\epr(H)$ is the (pointwise) downward closure of $\Lambda_H(\pSub{\epr}(H))$ in $\rpos^{2^E}$.

We refer to the gauge on the submodular realisable body as  the \emph{submodular gauge}, namely,
\[
\gauge^\Sub_{\epr,H}(q):=
\gauge_{K^\Sub_{\epr}(H)}(q)
=
\inf\{\alpha\ge 0\mid q\in \alpha\cdot K^\Sub_{\epr}(H)\}.
\]
We will primarily be interested in the unit edge profile $\one$, for this we write
$\gauge_H^\Sub:=\gauge^\Sub_{\one,H}$.
\end{definition}

While our results primarily rely on the unit profile $\one$ (matching the $\pSub{\one} $ constraint), we believe it beneficial to establish the geometric results for arbitrary edge profiles $\epr$. On the one hand, this illustrates that there is a robust principle underlying these results. On the other hand, edge profiles in a sense express constraints on the data, similar to the study of \emph{degree-aware} submodular width~\cite{DBLP:conf/pods/Khamis0S17,panda25}. Thus, this more general technical development may be useful for future work along those lines.

\begin{proposition}
\label{prop:body-geometry}
For every edge profile $\epr$, the set $K^\Sub_{\epr}(H)$ is nonempty, compact,
convex, and coordinatewise downward closed.
\end{proposition}

\begin{proof}
The zero vector belongs to $\pSub{\epr}(H)$, so $0\in K^\Sub_{\epr}(H)$.
Downward closure holds by definition.

For convexity, let $q_i\le \Lambda_H(b_i)$ with $b_i\in \pSub{\epr}(H)$ for
$i=1,2$, and let $0\le \alpha\le 1$. Then
\[
\alpha q_1+(1-\alpha)q_2
\le
\alpha \Lambda_H(b_1)+(1-\alpha)\Lambda_H(b_2)
=
\Lambda_H(\alpha b_1+(1-\alpha)b_2).
\]
Since $\pSub{\epr}(H)$ is convex by \Cref{lem:sub-slice-geometry}, also
$\alpha b_1+(1-\alpha)b_2\in \pSub{\epr}(H)$,
and hence $\alpha q_1+(1-\alpha)q_2\in K^\Sub_{\epr}(H)$. That is, $K^\Sub_\epr(H)$ is convex as well.

For compactness, \Cref{lem:sub-slice-geometry} states that
$\pSub{\epr}(H)$ is compact. Since $\Lambda_H$ is linear, hence continuous, the image
$U:=\Lambda_H(\pSub{\epr}(H))$
is compact in $\mathbb R_{\ge 0}^{2^E}$.

Because $U$ is compact it is also bounded. That is, there exists
$B<\infty$ such that
$u(F)\le B$ for all $u\in U$ and all $F\subseteq E(H)$.
Now consider
\[
S:=
\{(q,u)\in \rpos^{2^E}\times U\mid q\le u\}.
\]
Since $q\le u$ coordinatewise and every coordinate of every $u\in U$ is at most $B$,
we have
\[
S\subseteq [0,B]^{2^E}\times U.
\]
Moreover, $S$ is closed in $\mathbb R^{2^E}\times \mathbb R^{2^E}$, hence also closed in
the compact set $[0,B]^{2^E}\times U$. Therefore $S$ is compact.
Its projection onto the first coordinate is exactly $K^\Sub_{\epr}(H)$, so
$K^\Sub_{\epr}(H)$ is compact.
\end{proof}

The next lemma isolates the abstract variational mechanism behind our geometric reformulation. 
This part of the argument is not specific to submodularity. It holds for any compact family of feasible symmetric profiles, that feasibility is preserved under decreasing coordinates, and that the objective under consideration is monotone and positively homogeneous of degree $1$. 
In our application of the lemma below, the feasible family is the set of boundary lifts realised by functions in $\pSub{\epr}(H)$, and its downward closure is exactly the realisable body $K^\Sub_{\epr}(H)$. 
We nevertheless state and prove the following lemma in this general form to make it clear that this applies all the same to analogous width measures, and in particular to adaptive width via the realisable body $K^\Mod_{\epr}(H)$. 

\begin{lemma}
\label{lem:monotone-homogeneous-variational}
Let $K\subseteq \Sym(H)$ be nonempty and compact, and define its symmetric downward closure 
\[
K^{\downarrow\Sym}
:=
(\downarrow K)\cap \Sym(H)
=
\{q\in \Sym(H)\mid \exists u\in K\text{ such that }q\le u\}.
\]
That is, $K^{\downarrow\Sym}$ is the symmetric slice of the downward closure of $K$.
Let $\Phi:\Sym(H)\to \rpos$ be monotone and degree-$1$ positively homogeneous.
Then
\[
\sup_{u\in K}\Phi(u)
=
\sup_{q\in \Sym(H)}
\frac{\Phi(q)}{\gauge_{K^{\downarrow\Sym}}(q)},
\]
with the conventions $0/0:=0$ and $a/\infty:=0$ for $a<\infty$.
\end{lemma}

\begin{proof}
Since every $q\in K^{\downarrow\Sym}$ is dominated by some $u\in K$, monotonicity of
$\Phi$ gives
\[
\sup_{q\in K^{\downarrow\Sym}}\Phi(q)=\sup_{u\in K}\Phi(u).
\]

Now fix $q\in \Sym(H)$. If $\gauge_{K^{\downarrow\Sym}}(q)=\infty$, then the ratio is
$0$ by convention. So suppose
\(
\gamma:=\gauge_{K^{\downarrow\Sym}}(q)<\infty.
\)
Then for every $\alpha>\gamma$ one has $q\in \alpha K^{\downarrow\Sym}$, so
$q=\alpha q'$ for some $q'\in K^{\downarrow\Sym}$. Hence
\[
\Phi(q)=\alpha\,\Phi(q')
\le
\alpha \sup_{u\in K^{\downarrow\Sym}}\Phi(u)
=
\alpha \sup_{u\in K}\Phi(u).
\]
That is, we have $\Phi(q) \le \alpha \sup_{u \in K}\Phi(u)$ and dividing by $\alpha>\gamma\ge 0$ we have that
$\frac{\Phi(q)}{\alpha}\le \sup_{u\in K}\Phi(u)$.

If $\gamma>0$, taking the infimum over $\alpha>\gamma$ yields
\[
\frac{\Phi(q)}{\gamma}\le \sup_{u\in K}\Phi(u).
\]
If $\gamma=0$, then the same inequality holds for every $\alpha>0$, and taking the (right-sided) limit as $\alpha \to 0_+$ gives $\Phi(q)=0$, so the conclusion remains valid with
the convention $0/0:=0$.

For the reverse inequality, every $u\in K$ lies in $K^{\downarrow\Sym}$, so
$\gauge_{K^{\downarrow\Sym}}(u)\le 1$. Therefore for every $u\in K\setminus\{0\}$,
\[
\frac{\Phi(u)}{\gauge_{K^{\downarrow\Sym}}(u)}\ge \Phi(u).
\]
If $u=0$, then $\Phi(0)=0$ by positive homogeneity, so the same conclusion is
consistent with $0/0:=0$. Taking the supremum over $u\in K$ proves the reverse
inequality.
\end{proof}

We are missing only one final piece of the puzzle.
Positive homogeneity is essential in \Cref{lem:monotone-homogeneous-variational}. It is what allows the scaling parameter from the gauge to appear as the normalising denominator, in particular in the step passing from $\Phi(q)\le \alpha \sup_{u\in K}\Phi(u)$ for all $\alpha>\gamma$ to $\Phi(q)/\gamma\le \sup_{u\in K}\Phi(u)$. Fortunately, $\disp_H$ indeed has this property.

\begin{lemma}
\label{lem:disp-monotone-homogeneous}
 $\disp_H$ is monotone and positively
homogeneous of degree $1$.
\end{lemma}

\begin{proof}
If $q\le q'$, then every branch decomposition has $q$-width at most its
$q'$-width, hence $\disp_H(q)\le \disp_H(q')$. Also, for every $\alpha\ge 0$,
every branch decomposition has $(\alpha q)$-width equal to $\alpha$ times its
$q$-width, so $\disp_H(\alpha q)=\alpha\,\disp_H(q)$.
\end{proof}

\begin{theorem}[variational characterisation]
\label{thm:variational}
For every hypergraph $H$,
\[
\subwlift(H)=
\sup_{q\in \Sym(H)}
\frac{\disp_H(q)}{\gauge_H^\Sub(q)},
\]
with the convention that profiles with $\gauge_H^\Sub(q)=+\infty$ contribute
ratio $0$.
\end{theorem}

\begin{proof}
Let $K=\Lambda_H(\pSub{\one}(H))$.

By \Cref{prop:lift.sym}, the set $K$ is contained in $\Sym(H)$. Since $\Lambda_H$ is continuous and $\pSub{\one}(H)$ is compact by
\Cref{lem:sub-slice-geometry}, the set $K$ is compact.

Its symmetric downward closure is
\[
K^{\downarrow\Sym}
=
\{q\in \Sym(H)\mid \exists u\in K,\ q\le u\}
=
K^\Sub_{\one}(H)\cap \Sym(H).
\]
Since $K^\Sub_{\one}(H)$ is downward closed, its gauge agrees on symmetric
profiles with the gauge of its symmetric slice.
Indeed, one direction is immediate because
$K^\Sub_{\one}(H)\cap \Sym(H)\subseteq K^\Sub_{\one}(H)$.

Conversely, if $q\in \Sym(H)$ and $q\in \alpha K^\Sub_{\one}(H)$ with $\alpha>0$,
then $q/\alpha\in K^\Sub_{\one}(H)\cap \Sym(H)$.
If $\alpha=0$, then necessarily $q=\bf 0$ (recall that if $\mathbf{0}\in K$, then $\gauge_K(\mathbf{0})=0$).
Hence
\[
\gauge_{K^{\downarrow\Sym}}(q)=\gauge_H^\Sub(q)
\qquad\forall q\in \Sym(H).
\]

The functional $\disp_H$ is monotone and degree-$1$ positively homogeneous on
$\Sym(H)$ by \Cref{lem:disp-monotone-homogeneous}. Applying
\Cref{lem:monotone-homogeneous-variational} with $\Phi=\disp_H$ therefore gives
\[
\sup_{u\in K}\disp_H(u)
=
\sup_{q\in \Sym(H)}
\frac{\disp_H(q)}{\gauge_H^\Sub(q)}.
\]
By \Cref{def:lifted-width}, the left-hand side is exactly $\subwlift(H)$.
\end{proof}

By combining \Cref{thm:variational} with
\Cref{cor:lifted-bw}, one immediately obtains a variational approximation of $\subw(H)$ up to the factor $\tfrac32$. One can also observe that the supremum in \Cref{thm:variational} is in fact a maximum.

\Cref{thm:variational} is the starting point for
 later sections. It replaces direct reasoning about witnesses
$b\in \pSub{\one}$ by the geometry of the associated symmetric boundary profiles on $2^E$. This is also where our move to lifted parameters and cut weightings pays off. The variational characterisation provides us with new technical freedom, in the sense that we can argue lower bounds by arguing over arbitrary functions in $\Sym(H)$. We provide two examples to illustrate these new possibilities. The following sections then employ this in a more general way, where we construct useful $q\in \Sym(H)$ without consideration for their realisability in terms of functions in $\pSub{\one}(H)$.

\begin{example}

Let $H$ have edges $e_1,e_2,e_3,e_4$, all containing a common vertex $v$, and suppose that these are the only intersections among the edges. 
Then for every nontrivial proper set $\emptyset \subsetneq F \subsetneq E(H)$,
\(
\partial_H(F)=\{v\}.
\)
Hence for every $b\in \pSub{\one}(H)$,
\[
\Lambda_H b(F)=
\begin{cases}
0, & F\in \{\emptyset,E(H)\},\\[1mm]
b(\{v\}), & \emptyset \subsetneq F \subsetneq E(H).
\end{cases}
\]
Since $\{v\}\subseteq e_i$ for every $i$, the edge-domination constraint implies $b(\{v\})\le 1$. 
Thus every boundary lift considered in the definition of $\subwlift$ is constant on all nontrivial cuts, with value at most $1$.

Now define
\[
q(F):=|F|\cdot(4-|F|)\qquad\forall F\subseteq E(H).
\]
Then $q\in \Sym(H)$, with value $3$ on singleton cuts and value $4$ on cuts with $2$ edges on each side. 
We see that $q$ is not itself realisable, i.e., $q \not \in K^\Sub_\one(H)$. In fact, $\gauge_H^\Sub(q)=4$:
if $q\in \alpha K^\Sub_{\one}(H)$, then every nontrivial cut value of $q$ must be at most $\alpha$, so necessarily $\alpha\ge 4$. 
Conversely, equality is attained because $q\le 4\Lambda_H b$, where $b$ is the modular function given by $b(X)=\1[v\in X]$.
We see that symmetry is a far weaker requirement than submodular realisability. 
\end{example}

\begin{example}

\newcommand{\bal}{\mathsf{bal}}

The variational principle is useful even with profiles that are not themselves

realisable. Let $H=K_n$ be the complete graph on $n$ vertices, viewed as a
hypergraph with $2$-element edges, and define the following $q \in \Sym(H)$:
\[
q(F):=
\begin{cases}
1& \text{if } \min\{|F|,|E(H)\setminus F|\}\ge |E(H)|/3,\\
0& \text{otherwise}.
\end{cases}
\]
Observe that $\disp_H(q)=1$ as every branch decomposition of $E(H)$ has a cut with both sides
containing at least one third of the leaves, so some cut has
$q$-value $1$, while trivially every cut has value at most $1$.

Now let $b(X):=|X|/2$. Since every edge of $K_n$ has size $2$, we have $b\in \pMod{\one}(H)\subseteq \pSub{\one}(H)$, and hence $\Lambda_H b(F)=\frac{|\partial_H(F)|}{2}$.

Now observe that for every $F\subseteq E(H)$ with
$\min\{|F|,|E(H)\setminus F|\}\ge |E(H)|/3$, one has
$|\partial_H(F)|\ge \frac n6$.

Write $U:=V(H)\setminus \partial_H(F)$. Then every vertex of $U$ has all its
incident edges contained either in $F$ or in $E(H)\setminus F$. Since
$K_n[U]$ is complete, all vertices of $U$ must choose the same side, for
otherwise an edge of $K_n[U]$ would belong to both $F$ and $E(H)\setminus F$,
a contradiction. That is, one side of the partition is covered by $\partial_H(F)$. As both $F$ and $E(H)\setminus F$ have size at least $|E(H)|/3$, this yields
\(
|\partial_H(F)|(n-1)\ge \frac{|E(H)|}{3}=\frac{n(n-1)}{6},
\)
so $|\partial_H(F)|\ge n/6$.

Thus on every balanced cut,
$\Lambda_H b(F)=\frac{|\partial_H(F)|}{2}\ge \frac n{12}$, and equivalently,
\(
q\le \frac{12}{n}\,\Lambda_H b.
\)
That in particular means that $q$ has very small submodular gauge, namely $\gauge_H^\Sub(q)\le \frac{12}{n}$. That is, $q$ lies deep inside the realisable body $K^\Sub_\one(H)$. Intuitively, this implies it can be scaled up to a ``heavier'' realisable cut weighting. Looking at the definition of $q$ it is not clear how, but the convex nature of the induced geometry guarantees that it is.

To conclude, combining the low gauge with $\disp_H(q)=1$ gives
\[
\subwlift(H)\ge
\frac{\disp_H(q)}{\gauge_H^\Sub(q)}
\ge \frac{n}{12}.
\]
Note that this is slightly below the actual submodular width $n/2$ of the clique. The point we wish to illustrate is how the variational perspective opens up entirely new approaches to proving lower bounds on $\subw$.
\end{example}

\subsection{An Antiblocker-Dual Interpretation}
\label{sec:geom.duality}

The quantity $\gauge^\Sub_{\one,H}(q)$ asks for the smallest scaling factor by which a candidate cut weighting $q$ must be divided before it becomes realisable by some submodular witness. 
One can then consider a pricing scheme on cuts as a kind of dual certificate for this question. That is, the  scheme assigns a price $w(F)$ to each bipartition $(F,E\setminus F)$ and is valid if every realisable profile collects total price at most $1$. 
Because the realisable body is coordinatewise downward closed inside the nonnegative orthant, this is precisely the setting of so-called \emph{antiblocking duality} for convex corners. Antiblockers originate in combinatorial optimisation \cite{fulkerson1972anti} and also play an important role in links between graph theory and information theory, for example in classical work of Csiszár et al. \cite{DBLP:journals/combinatorica/CsiszarKLMS90}.
 Intuitively, an antiblocker consists of all cut-pricing schemes that charge at most $1$ to every realisable profile, and the gauge of $q$ is the maximum total charge that any such scheme can extract from $q$.

Formally, for any set $K\subseteq \rpos^{2^E}$, define its \emph{antiblocker} by
\[
\operatorname{abl}(K)
:=
\Bigl\{
w\in \rpos^{2^E}
\Bigm|
\langle q,w\rangle \le 1
\ \text{for all } q\in K
\Bigr\},
\]
where
\[
\langle q,w\rangle := \sum_{F\subseteq E} q(F)w(F)
\]
denotes the standard inner product on $\mathbb R^{2^E}$.

Thus $\operatorname{abl}(K)$ consists of all nonnegative linear functionals
which are bounded by $1$ on $K$. In our case, this gives a dual description of
the submodular gauge.

\begin{proposition}
\label{prop:support-separation}
Let $K\subseteq \mathbb R^d$ be nonempty, compact, and convex, and for $y\in \mathbb R^d$ define
$h_K(y):=\max_{x\in K}\langle x,y\rangle$ .
Then,
for every $\alpha\ge 0$,
\[
x\in \alpha K
\quad\Longleftrightarrow\quad
\langle x,y\rangle \le \alpha\, h_K(y)
\ \text{for all } y\in \mathbb R^d.
\]
\end{proposition}

\begin{proof}
The forward implication is immediate: if $x=\alpha z$ with $z\in K$, then for
every $y\in \mathbb R^d$,
\[
\langle x,y\rangle
=
\alpha \langle z,y\rangle
\le
\alpha\, h_K(y).
\]

Conversely, suppose $x\notin \alpha K$. Since $\alpha K$ is compact and convex,
hence closed and convex, Rockafellar~\cite[Theorem~11.5]{rockafellar1997convex}
implies that $\alpha K$ is the intersection of the closed half-spaces
containing it. Therefore there exists a closed half-space
$H=\{z\in \mathbb R^d \mid \langle z,y\rangle\le \beta\}$
such that $\alpha K\subseteq H$ but $x\notin H$. Then
$h_{\alpha K}(y)=\max_{z\in \alpha K}\langle z,y\rangle \le \beta < \langle x,y\rangle$.
Since $h_{\alpha K}(y)=\alpha\, h_K(y)$, it follows that
\[
\langle x,y\rangle>\alpha\, h_K(y).
\]
This proves the contrapositive of the reverse implication.
\end{proof}

\begin{theorem}
\label{thm:gauge-antiblocker}
For every hypergraph $H$, every edge profile $\epr$, and every
$q\in \rpos^{2^E}$,
\[
\gauge^\Sub_{\epr,H}(q)
=
\sup_{w\in \operatorname{abl}(K^\Sub_{\epr}(H))}
\langle q,w\rangle,
\]
with the convention that the right-hand side may be $+\infty$.
\end{theorem}

\begin{proof}
Write $K:=K^\Sub_{\epr}(H)$. By \Cref{prop:body-geometry}, the set $K$ is
nonempty, compact, convex, and coordinatewise downward closed.

For $y\in \mathbb R^{2^E}$, write
$h_K(y):=\max_{x\in K}\langle x,y\rangle$ as in \Cref{prop:support-separation}, and
observe that $\operatorname{abl}(K)=\{w\in \rpos^{2^E}: h_K(w)\le 1\}$.

Define $y_+(F):=\max\{y(F),0\}$. We first claim that
$h_K(y)=h_K(y_+)$ for all $y\in \mathbb R^{2^E}$.
First, for $x\in K$, define $x'\in \rpos^{2^E}$ by
\[
x'(F):=
\begin{cases}
x(F), & y(F)\ge 0,\\
0, & y(F)<0.
\end{cases}
\]
Then $0\le x'\le x$ coordinatewise, so $x'\in K$, and
$\langle x,y_+\rangle=\langle x',y\rangle\le h_K(y)$.
Taking the maximum over $x\in K$ gives $h_K(y_+)\le h_K(y)$. The reverse
inequality is immediate from $y\le y_+$ and $K \subseteq \rpos^{2^E}$:
since every $x \in K$ is nonnegative, also $\langle x,y\rangle \le \langle x, y_+ \rangle$ for every $x$.

Now fix $\alpha\ge 0$ and $q\in \rpos^{2^E}$. By
\Cref{prop:support-separation},
\[
q\in \alpha K
\quad\Longleftrightarrow\quad
\langle q,y\rangle\le \alpha\, h_K(y)
\ \text{for all } y\in \mathbb R^{2^E}.
\]
Using $q\ge 0$ and $h_K(y)=h_K(y_+)$, this is equivalent to $\langle q,u\rangle\le \alpha\, h_K(u)$ for all $u\in \rpos^{2^E}$.

We claim that this is in turn equivalent to
$\langle q,w\rangle\le \alpha$
for all  $w\in \operatorname{abl}(K)$.
If $\langle q,u\rangle\le \alpha h_K(u)$ for all $u\in \rpos^{2^E}$ and
$w\in \operatorname{abl}(K)$, then $h_K(w)\le 1$, so
$\langle q,w\rangle\le \alpha\, h_K(w)\le \alpha$.

Conversely, assume $\langle q,w\rangle\le \alpha$ for all $w\in \operatorname{abl}(K)$, and fix $u\in \rpos^{2^E}$.

If $h_K(u)=0$, then $tu\in \operatorname{abl}(K)$ for every $t>0$, hence
$t\langle q,u\rangle=\langle q,tu\rangle\le \alpha$ for all $t>0$,
which forces $\langle q,u\rangle=0=\alpha h_K(u)$.

If $h_K(u)>0$, then $u/h_K(u)\in \operatorname{abl}(K)$, so
\[
\frac{\langle q,u\rangle}{h_K(u)}
=
\Bigl\langle q,\frac{u}{h_K(u)}\Bigr\rangle
\le \alpha,
\]
that is, $\langle q,u\rangle\le \alpha\, h_K(u)$.

We have therefore shown that, for every $\alpha\ge 0$,
\[
q\in \alpha K
\quad\Longleftrightarrow\quad
\langle q,w\rangle\le \alpha
\ \text{for all } w\in \operatorname{abl}(K).
\]
Hence
\[
\gauge^\Sub_{\epr,H}(q)
=
\inf\{\alpha\ge 0:\langle q,w\rangle\le \alpha
\ \text{for all } w\in \operatorname{abl}(K)\}
=
\sup_{w\in \operatorname{abl}(K)}\langle q,w\rangle.
\]
\end{proof}

Combining this with \Cref{thm:variational} yields an equivalent dual form of
the lifted submodular width.

\begin{corollary}[Dual variational form]
\label{cor:subwlift-dual}
For every hypergraph $H$,
\[
\subwlift(H)
=
\sup_{q\in \Sym(H)}\ 
\inf_{w\in \operatorname{abl}(K^\Sub_{\one}(H))}
\frac{\disp_H(q)}{\langle q,w\rangle},
\]
with the conventions that $a/0:=+\infty$ for $a>0$, and $0/0:=0$.

\end{corollary}

\begin{proof}
By \Cref{thm:variational} and \Cref{thm:gauge-antiblocker},
\[
\subwlift(H)
=
\sup_{q\in \Sym(H)}
\frac{\disp_H(q)}
{\sup_{w\in \operatorname{abl}(K^\Sub_{\one}(H))}\langle q,w\rangle}.
\]
For fixed $q$, the quantity $\disp_H(q)$ is independent of $w$, while
$\langle q,w\rangle\ge 0$. Hence
\[
\frac{\disp_H(q)}
{\sup_{w\in \operatorname{abl}(K^\Sub_{\one}(H))}\langle q,w\rangle}
=
\inf_{w\in \operatorname{abl}(K^\Sub_{\one}(H))}
\frac{\disp_H(q)}{\langle q,w\rangle},
\]
with the stated conventions. Taking the supremum over $q\in \Sym(H)$ gives the
result.
\end{proof}

We have thus obtained a dual description of the submodular gauge in terms of  antiblockers. 
Although we do not make further use of this perspective in the present paper, it seems to us to provide a helpful alternative interpretation of the gauge. Instead of reasoning directly about realisability, one may test a profile against nonnegative cut-pricing schemes. We expect that this dual viewpoint may be useful both for proving lower bounds on $\gauge_H^\Sub$ and for showing tightness of upper bounds coming from explicit realisations.

\section{General Cut Certificates from Treewidth}
\label{sec:support-inputs}
\providecommand{\load}{\operatorname{load}}

The variational formula in \Cref{sec:geometry} reduces lower bounds on
$\subwlift(H)$, and hence on $\subw(H)$, to constructing a symmetric edge-set
function $q\in \Sym(H)$ with two properties.
First, $q$ must be large on all sufficiently balanced cuts of $E(H)$.
Second, $q$ must admit a concrete representation as a cut profile that is induced by an edge profile $\alpha$ on the line graph $M:=P(H^*)$, defined as follows:
\[
q(F)=\sum_{e\in \cut{M}(F)}\alpha_e
\qquad \forall F\subseteq E(H).
\]
Moreover, we require the local incident load of $\alpha$ on the line graph to be strictly controlled, so that the resulting cut profile can ultimately be realised with a small submodular gauge cost.

In this section, we develop the core routing machinery to achieve the first half of this goal. We show that the treewidth of the line graph guarantees the existence of a balanced cut profile with a bounded local load. Later, in \Cref{sec:applications}, we will establish the structural interfaces that allow us to efficiently translate this bounded line-graph load into a bounded submodular gauge on the hypergraph.

\subsection{Balanced Edge Set Functions Force Large Cuts}

Throughout this section, a probability measure on a finite set $\Omega$ is
simply a weight function $\rho:\Omega\to \rpos$ with
$\sum_{x\in \Omega}\rho(x)=1$, extended additively to subsets by
\[
\rho(F):=\sum_{x\in F}\rho(x)
\qquad \forall F \subseteq \Omega.
\]
In particular, saying that a point $x\in \Omega$ has $\rho$-mass at most
$1/2$ means $\rho(\{x\})\le 1/2$. The following is an adaptation of standard arguments for balanced cuts in branch decompositions, adapted to the specifics we use later.

\begin{lemma}
\label{lem:branch-one-third}
Let $T$ be a subcubic tree with leaf set $\Omega$, and let $\rho$ be a
probability measure on $\Omega$. Assume that every leaf $\ell \in \Omega$ has $\rho$-mass at most $1/2$.
Then some edge $f$ of $T$ induces a bipartition
$(X_f,\Omega\setminus X_f)$ of $\Omega$ such that
\(
\min\{\rho(X_f),\rho(\Omega\setminus X_f)\}\ge 1/3.
\)
\end{lemma}

\begin{proof}
Assign weight $\rho(\ell)$ to each leaf $\ell$ and weight $0$ to every internal
vertex. Start at an arbitrary vertex $t_0$. If some component $C$ of $T-t_0$
has total weight greater than $1/2$, move to the unique neighbor $t_1$ of
$t_0$ that lies in $C$, and continue similarly. This process is well defined,
since two distinct components of $T-t$ cannot both have weight greater than
$1/2$. Moreover, if we move from $t_i$ into a component $C_i$ of $T-t_i$ with
weight greater than $1/2$, then in $T-t_{i+1}$ the component containing $t_i$
has weight $1-\rho(C_i)<1/2$, so the process cannot immediately return across
the same edge. Since $T$ is acyclic, a walk that never immediately
backtracks also never revisits a vertex, so the process must terminate. Hence it
terminates at a vertex $t^\ast$ such that every component of $T-t^\ast$ has
total weight at most $1/2$.

If $t^\ast$ is a leaf, then the unique component of $T-t^\ast$ has weight
$1-\rho(\{t^\ast\})\le 1/2$, so $\rho(\{t^\ast\})\ge 1/2$. By hypothesis
$\rho(\{t^\ast\})\le 1/2$, hence $\rho(\{t^\ast\})=1/2$, and the incident edge
gives a $1/2$--$1/2$ cut.

Assume now that $t^\ast$ is internal. Since $T$ is subcubic,
$\deg_T(t^\ast)\in\{2,3\}$. If the component masses are $a,b$ in the
degree-$2$ case, then $a+b=1$ and $a,b\le 1/2$, so $a=b=1/2$. If the component
masses are $a,b,c$ in the degree-$3$ case, then $a+b+c=1$ and each of $a,b,c$
is at most $1/2$, so not all three can be less than $1/3$. Thus in every case
some incident edge induces a cut with both sides of mass at least $1/3$.
\end{proof}

\begin{lemma}
\label{lem:general-balanced-profile}
Let $q \in \Sym(H)$, let $\rho$ be a probability measure on $E$ such that
$\rho(\{e\})\le 1/2$ for all $e\in E$, and let
$\phi:[0,1/2]\to \rpos$ be nondecreasing. If for all $F \subseteq E$ it holds that
$$q(F)\ge \phi\bigl(\min\{\rho(F),\rho(E\setminus F)\}\bigr),$$
then $\disp_H(q)\ge \phi(1/3)$.
\end{lemma}

\begin{proof}
Let $(T,\delta)$ be any branch decomposition of $H$. By
\Cref{lem:branch-one-third}, some cut $f\in E(T)$ satisfies
\(
\min\{\rho(X_f),\rho(E\setminus X_f)\}\ge \frac13.
\)
Hence $q(X_f)\ge \phi(1/3)$ for some cut in every branch decomposition of $H$.
\end{proof}

\subsection{From Treewidth to Multicommodity Flows}

We now turn to the main technical step, constructing a symmetric cut weighting $q\in \Sym(H)$ that is simultaneously large on balanced cuts and induced by an edge profile on the line graph with small local incident load. To obtain such a profile, we follow an established path from treewidth, to well-linked sets, to multicommodity flow routing.
Accordingly, this subsection is not a new routing theorem, but rather a tailored adaptation of results of \citet{Reed97}, \citet{FeigeHajiaghayiLee08}, and \citet{ChekuriKhannaShepherd05}, adapted to the particular form needed for our variational framework.

From now on, let $M:=P(H^*)$ be the line graph of $H$. Thus every set $X\subseteq V(M)$ is identified with a subset $F\subseteq E(H)$. For such
$F$ write
\(
\cut{M}(F):=\{xy\in E(M) \mid  x\in F,\ y\notin F\}
\)
for the set of edges in $M$ that cross the cut.

Let $\alpha\in \rpos^{E(M)}$ be an edge profile on $M$. Define
$q_{M,\alpha}(F):=\sum_{e\in \cut{M}(F)}\alpha_e$
for all $F\subseteq E(H)$.
Note that $q_{M,\alpha}\in \Sym(H)$, since
$\cut{M}(F)=\cut{M}(E(H)\setminus F)$. 

For $h\in E(H)=V(M)$, the total $\alpha$-weight incident with $h$ in $M$ is
$\alpha(\Inc_M(h))$.

Let $\mathcal P(M)$ denote the set of simple paths in $M$. A
\emph{path weighting} on $M$ is a function
$\omega:\mathcal P(M)\to  \rpos$
with total mass $\sum_{P\in\mathcal P(M)}\omega(P)=1$. It induces an edge profile $\alpha^\omega$ in the line graph $M$ with weights
\[
\alpha^\omega_e:=\sum_{P\ni e}\omega(P)
\qquad \forall e\in E(M),
\]
and with it the symmetric edge set function
$q_\omega:=q_{M,\alpha^\omega}\in \Sym(H)$.

\begin{lemma}
\label{prop:edge-usage-in-cut-cone}
For every path weighting $\omega$ on $M$ and every $F\subseteq E(H)$,
\[
q_\omega(F)
=
\sum_{P\in\mathcal P(M)} \omega(P)\,\abs{E(P)\cap\cut{M}(F)}.
\]
\end{lemma}

\begin{proof}
By the definition of $\alpha^\omega$ and by interchanging two finite sums,
\[
q_\omega(F)
=\sum_{e\in\cut{M}(F)}\sum_{P\ni e}\omega(P)
=\sum_{P\in\mathcal P(M)}\omega(P)\,\abs{E(P)\cap\cut{M}(F)}.
\]
\end{proof}

The point of the path-weighting formalism is that $q_\omega(F)$ is exactly the
expected number of edges of $\cut{M}(F)$ traversed by a random simple path drawn
from $\omega$. Thus, if $\omega$ spreads mass uniformly over many terminal
pairs, then every cut separating many of those terminals forces $q_\omega(F)$ to
be large. At the same time, when $\omega$ comes from a feasible
node-capacitated multicommodity flow, the local incident load of the induced edge
profile $\alpha^\omega$ is controlled by the vertex capacities. This is the
mechanism that converts treewidth into a balanced cut profile induced by an edge weighting on the line graph.

Let $G$ be a graph.
A set $U\subseteq V(G)$ is \emph{node-well-linked} if for every $A,B\subseteq U$ with
$|A|=|B|$, there exist $|A|$ pairwise internally vertex-disjoint $A$--$B$ paths
in $G$. When $A\cap B\neq\emptyset$, we
allow trivial paths joining each vertex of $A\cap B$ to itself.

\begin{proposition}[Reed~\cite{Reed97}, see also Lemma~2.1~\cite{ChekuriChuzhoy13}]
\label{prop:treewidth-node-well-linked-set}
Let $G$ be a graph. Let $k$ be the size of the largest node-well-linked set in $G$. Then
$k \le \tw(G) \le 4k$.
\end{proposition}

For $S\subseteq V(G)$, write
$N_G(S):=\{x\in V(G)\setminus S \mid \exists y \in S \text{ s.t. } xy \in E(G)\}$.
A set $U\subseteq V(G)$ is \emph{$1$-node-cut-linked} if, for every
$S\subseteq V(G)$ with $|S\cap U|\le |U|/2$, one has $|N_G(S)|\ge |S\cap U|$.

It is a standard consequence of Menger's Theorem~\cite{menger1927allgemeinen} that node-well-linked sets are also 1-node-cut-linked (see also~\cite[Section 1.3]{ChekuriKhannaShepherd05}).

\begin{proposition}
\label{prop:well-linked-node-cut-linked}
Every node-well-linked set is also $1$-node-cut-linked.
\end{proposition}

Our goal will be to move from cut-linkedness to a technically pleasant function in $\Sym(H)$. We do so through node-capacitated multicommodity flows. 

For $\{u,v\}\in\binom{U}{2}$ let $\mathcal P_{uv}(G)$ denote the set of simple
$u$--$v$ paths in $G$.

Let $G$ be a graph and let $U\subseteq V(G)$. 
A \emph{node-capacitated multicommodity flow} on $U$ is a family
\[
f=\bigl(f_{uv}\bigr)_{\{u,v\}\in\binom{U}{2}}
\qquad\text{with}\qquad
f_{uv}:\mathcal P_{uv}(G)\to \mathbb R_{\ge 0}.
\]
The value of the $\{u,v\}$-commodity is
\[
\|f_{uv}\|_1:=\sum_{P\in\mathcal P_{uv}(G)} f_{uv}(P).
\]
For $x\in V(G)$, the internal load of $x$ is
\[
\load_f(x):=
\sum_{\{u,v\}\in\binom{U}{2}}
\sum_{\substack{P\in\mathcal P_{uv}(G)\\ x\in V(P)\setminus\{u,v\}}}
f_{uv}(P).
\]
The flow is \emph{feasible} if $\load_f(x)\le 1$ for every $x\in V(G)$.

If $\|f_{uv}\|_1=\vartheta$ for every $\{u,v\}\in\binom{U}{2}$, we say that
$f$ routes $\vartheta$ units between every unordered terminal pair.

We use the convention that only internal vertices count against capacity; using
the alternative convention in which endpoints contribute weight $1/2$ changes
only constant differences.

\begin{proposition}[Feige, Hajiaghayi, and Lee,~Theorem 4.1~\cite{FeigeHajiaghayiLee08}]
\label{prop:fhl-product-node-flow-gap}
There exists an universal constant $C >0$ such that for every graph
$G$, every terminal set $U\subseteq V(G)$ with $|U|=k\ge 2$, and every product
multicommodity demand supported on $U$, the node-capacitated max-flow/min-cut gap is
at most $C \log(k)$.
\end{proposition}

\begin{proposition}
\label{prop:cks-uniform-node-dictionary}
There exists a universal constant $C>0$ such that the following holds.
Let $G$ be a graph, let $U\subseteq V(G)$ be a 1-node-cut-linked set with $m:=|U|\ge 2$. Then $G$
admits a feasible node-capacitated multicommodity flow routing at least
$1/(C m\log m)$ units between every unordered pair
$\{u,v\}\in\binom{U}{2}$.
\end{proposition}

\begin{proof}
The result is due to Chekuri, Khanna, and Shepherd \cite[Section~1.3]{ChekuriKhannaShepherd05}. They use the
following terminology. For a nonnegative weight function
$\pi:X\to \rpos$ on terminals, $X$ is $\pi$-flow-linked if the
product demand
\[
\frac{\pi(u)\pi(v)}{\pi(X)}
\]
can be feasibly routed between every unordered terminal pair $u,v\in X$ in the
relevant edge- or node-capacitated sense. In the node case, $X$ is
$\pi$-node-cut-linked if
\[
|N_G(S)|\ge \pi(S\cap X)
\]
whenever $\pi(S\cap X)\le \pi(X)/2$. They state that if $X$ is
$\pi$-cut-linked, then it is $(\pi/\beta)$-flow-linked, where $\beta$ is the
worst-case max-flow/min-cut gap for product multicommodity flow instances in
the graph. By \Cref{prop:fhl-product-node-flow-gap}, the worst-case node-capacitated
max-flow/min-cut gap for product demands supported on $U$ is at most
$C \log m$.

Apply this with $X=U$ and $\pi(u)=1$ for every $u\in U$. Then, since $U$ is
$1$-node-cut-linked, it is $(\pi/\beta)$-flow-linked with
$\beta\le C\log m$. The corresponding product demand between each unordered
pair is
\[
\frac{(1/\beta)(1/\beta)}{m/\beta}
=
\frac{1}{\beta m}
\ge
\frac{1}{C m\log m},
\]
as claimed.
\end{proof}

\begin{lemma}
\label{prop:node-cut-linked-uniform-node-flow}
There exists a universal constant $C>0$ with the following
property. For every graph $G$ and every $1$-node-cut-linked set
$U\subseteq V(G)$ with $m:=|U|\ge2$, there is a feasible node-capacitated
multicommodity flow in $G$ that routes
\[
\vartheta\ge \frac{1}{C\,m\log(m)}
\]
units between every unordered pair $\{u,v\}\in\binom{U}{2}$.
\end{lemma}

\begin{proof}
This is exactly \Cref{prop:cks-uniform-node-dictionary}, after increasing the
universal constant if necessary.
\end{proof}

\begin{definition}
\label{def:normalised-path-weighting}
Let $U\subseteq E(H)$ with $|U|=m\ge 2$, and let
$f=\bigl(f_{uv}\bigr)_{\{u,v\}\in\binom{U}{2}}$
be a feasible node-capacitated multicommodity flow in $M$ such that
$\|f_{uv}\|_1=\vartheta$ for every $\{u,v\}\in\binom{U}{2}$.
Define the \emph{normalised path weighting}  $\omega_f$ as 
\[
\omega_f(P):=
\frac{1}{\binom{m}{2}\vartheta}
\sum_{\{u,v\}\in\binom{U}{2}} f_{uv}(P)
\qquad(P\in\mathcal P(M)),
\]
where $f_{uv}(P)=0$ for $P\notin \mathcal P_{uv}(M)$.
\end{definition}

\begin{lemma}
\label{lem:normalized-path-weighting-properties}
With notation as in \Cref{def:normalised-path-weighting}, $\omega_f$ is a normalised path weighting of total mass $1$.
Moreover:
\begin{enumerate}
\item for every $\{u,v\}\in\binom{U}{2}$, the total $\omega_f$-mass of
$u$--$v$ paths is exactly $1/\binom{m}{2}$;
\item if $\alpha^{\omega_f}$ is the induced edge weighting on the line graph, then
\[
\alpha^{\omega_f}(\Inc_M(h))
\le
\frac{4}{m}+\frac{2}{\binom{m}{2}\vartheta}
\qquad \forall h\in E(H),
\]
\item in particular, if $\vartheta\ge 1/(C\,m\log m)$, then
\[
\alpha^{\omega_f}(\Inc_M(h))\le C'\frac{\log m}{m}
\qquad \forall h\in E(H)
\]
for some constant $C'$ depending only on $C$.
\end{enumerate}
\end{lemma}

\begin{proof}
Write $\omega=\omega_f$. The total mass of $\omega$ is
\[
\sum_{P\in\mathcal P(M)}\omega(P)
=
\frac{1}{\binom{m}{2}\vartheta}
\sum_{\{u,v\}\in\binom{U}{2}}\sum_P f_{uv}(P)
=1.
\]
For a fixed terminal pair $\{u,v\}$, the total $\omega$-mass assigned to
$u$--$v$ paths is $\vartheta/(\binom{m}{2}\vartheta)=1/\binom{m}{2}$.

Fix $g\in E(H)$. The total $\omega$-mass of paths for which $g$ is an endpoint
of the terminal pair is at most $2/m$, with equality only when $g\in U$. For
internal use, let $\varphi_g(u,v)$ denote the total flow of the
$\{u,v\}$-commodity through $g$ as an internal vertex. Then
\[
\sum_{\substack{P\in\mathcal P(M)\\ g\text{ internal in }P}}\omega(P)
=
\frac{1}{\binom{m}{2}\vartheta}
\sum_{\{u,v\}\in \binom{U}{2}}\varphi_g(u,v).
\]
Feasibility of the node-capacitated flow gives
\(
\sum_{\{u,v\}\in \binom{U}{2}} \varphi_g(u,v)\le 1,
\)
so the total $\omega$-mass of paths that use $g$ internally is at most
$1/(\binom{m}{2}\vartheta)$. Since every path in the support is simple,
\[
\alpha^\omega(\Inc_M(g))
=
\sum_{P\in\mathcal P(M)}\omega(P)\deg_P(g)
\le
2\sum_{P\in\mathcal P(M)}\omega(P)\mathbf 1[g\in V(P)]
\le
\frac{4}{m}+\frac{2}{\binom{m}{2}\vartheta}.
\]
For the final statement, if $\vartheta \ge 1/(C m\log m)$, then
\[
\frac{2}{\binom{m}{2}\vartheta}
\le
\frac{4C\log m}{m-1}
\le
\frac{8C\log m}{m}.
\]
The remaining term $4/m$ is also at most $C''\log m/m$ for all $m\ge2$, after
choosing $C''$ large enough. This gives the desired constant $C'$.
\end{proof}

\subsection{From Multicommodity Flows to Cut Weightings}
We now extract from the routed flow a concrete symmetric cut weighting. The previous subsection produced a path weighting $\omega$ that spreads mass uniformly over many terminal pairs while keeping the induced incident load small. The next lemma converts the first property into a lower bound on cut values; the second property will later allow us to realise the resulting profile with small gauge cost.

\begin{lemma}
\label{prop:uniform-terminal-crossing}
Let $U\subseteq E(H)$ have size $m\ge 2$, and let $\rho$ be
\[
\rho(F):=\frac{|F\cap U|}{|U|} 
\qquad \forall F\subseteq E(H).
\]
Let $\omega$ be a path weighting on $M$ such that, for every unordered pair
$\{u,v\}\in\binom{U}{2}$, the total $\omega$-mass of $u$--$v$ paths is exactly
$1/\binom{m}{2}$. Then
\[
q_\omega(F)\ge 2\rho(F)\rho(E(H)\setminus F)
\qquad \forall F\subseteq E(H).
\]
\end{lemma}

\begin{proof}
Fix $F\subseteq E(H)$ and write
$a:=|U\cap F|$, $b:=|U\setminus F|=|U \cap (E(H)\setminus F)|$.
Every path whose endpoints lie on opposite sides of
$F$ and  $E(H)\setminus F$ contains at least one edge of $\cut{M}(F)$. Hence, by
\Cref{prop:edge-usage-in-cut-cone},
\[
q_\omega(F)\ge \frac{ab}{\binom{m}{2}}
=\frac{2ab}{m(m-1)}
\ge \frac{2ab}{m^2}
= \frac{2 |U \cap F| \cdot |U \cap (E(H)\setminus F)|}{|U|^2}
=2\rho(F)\rho(E(H)\setminus F).
\]
\end{proof}

\begin{lemma}\label{lem:treewidth-to-cut-profile}
There exists a universal constant $c_{\mathrm{sup}}>0$ with the following
property. Let $H$ be a hypergraph, let $M$ be the line graph of $H$, and
assume that $\tw(M)\ge 8$. Then there exist a subset $U\subseteq E(H)$ of
size $m:=|U|\ge \tw(M)/4$, and an edge profile $\alpha\in \mathbb R_{\ge 0}^{E(M)}$
such that, writing
\(
\rho(F):=\frac{|F\cap U|}{m}
\) for all $F \subseteq E(H)$, 
the induced symmetric edge set function
$q_{M,\alpha}\in \Sym(H)$
satisfies
\[
q_{M,\alpha}(F)\ge 2\rho(F)\rho(E(H)\setminus F)
\qquad \forall F\subseteq E(H),
\]
and
\[
\alpha(\Inc_M(e))\le c_{\mathrm{sup}}\,\frac{\log m}{ m}
\qquad \forall e\in E(H).
\]

In particular, $\rho$ is the uniform probability measure on $U$, and every
point of $U$ has $\rho$-mass exactly $1/m\le 1/2$.
\end{lemma}

\begin{proof}
By \Cref{prop:treewidth-node-well-linked-set}, there is a node-well-linked set
$U\subseteq V(M)=E(H)$ with $m:=|U|\ge \tw(M)/4$. Since $\tw(M)\ge 8$, we have $m\ge 2$.
Let $\rho$ be the uniform probability measure on $E(H)$ supported on $U$.

By \Cref{prop:well-linked-node-cut-linked}, the set $U$ is $1$-node-cut-linked.
Hence by \Cref{prop:node-cut-linked-uniform-node-flow} there exists a feasible
node-capacitated multicommodity flow routing
\[
\vartheta\ge \frac{1}{C\,m\log m}
\]
between every unordered pair in $\binom{U}{2}$.
Construct the normalized path weighting $\omega=\omega_f$ from this flow as in
\Cref{def:normalised-path-weighting}, and set $\alpha$ to be $\alpha^\omega$.
Then, by \Cref{lem:normalized-path-weighting-properties},
\[
\alpha(\Inc_M(h))\le c_{\mathrm{sup}}\frac{\log m}{m}
\qquad \forall h\in E(H).
\]
The balanced cut lower bound follows from \Cref{prop:uniform-terminal-crossing}.
Since $m\ge 2$ and $\rho$ is uniform on $U$, every point of $U$ has
$\rho$-mass at most $1/2$.
\end{proof}

\begin{theorem}
\label{thm:flow.engine}
There exists a universal constant $C>0$ such that the following holds.
Let $H$ be a hypergraph and let $M:=P(H^*)$.
If $\tw(M)\ge 8$, then there exists an edge profile $\alpha\in\mathbb R_{\ge0}^{E(M)}$
such that
\[
q_{M,\alpha}\in\Sym(H),\qquad
\disp_H(q_{M,\alpha})\ge \frac49,
\qquad\text{and}\qquad
\max_{e\in E(H)}\alpha(\Inc_M(e))\le C\,\frac{\log \tw(M)}{\tw(M) }.
\]
\end{theorem}
\begin{proof}
Apply \Cref{lem:treewidth-to-cut-profile} to obtain $U$, $m$, $\rho$, and $\alpha$ with
\[
q_{M,\alpha}(F)\ge 2\rho(F)\rho(E(H)\setminus F)
\qquad \forall F\subseteq E(H),
\]
and
\[
\alpha(\Inc_M(e))\le c\frac{\log m}{m} \le c' \frac{\log \tw(M)}{\tw(M)}
\qquad \forall e\in E(H).
\]
Now note that every point mass of $\rho$ from the lemma is at most $1/2$. It is now sufficient to apply \Cref{lem:general-balanced-profile} with
$\phi(x)=2x(1-x)$
to directly obtain $\disp_H(q_{M,\alpha})\ge \phi(1/3)=\frac49$.
\end{proof}

\section{Connecting $\ghw$ and $\subw$}
\label{sec:applications}

We now apply the machinery of the previous sections by showing how submodular width and $\ghw$ are related under various structural conditions.
We first present two natural structural conditions under which the submodular gauge for the $q$ from \Cref{thm:flow.engine} is in $O(\log \ghw(H)/\ghw(H))$. This then directly yields the inverse lower bound for $\subw(H)$ through the variational characterisation of $\subwlift$ (\Cref{thm:variational}) combined with the transfer from $\subwlift$ to $\subw$ (\Cref{cor:lifted-bw}).

As a starting point we recall a result by Lanzinger~\cite{LanzingerPODS22} that states that PTIME and FPT CQ evaluation collapse for degree 2 (assuming $W[1]\neq $ FPT). In a structural sense this implies that for every class of degree 2 hypergraphs $\mathcal H$, $\ghw$ is bounded if and only if $\subw$ is bounded, although the proof in \cite{LanzingerPODS22} does not produce any concrete bounds on how the measures are related. 

Here we show how to recover the complexity collapse through our framework, which also yields the first explicit bound relating the two measures in this class. In fact it turns out that the natural proof for the degree 2 case in our framework readily generalises to two (incomparable) significantly more general conditions. The strategy for the results in this section is simple, we demonstrate a general principle for converting edge profiles on the line graph $M=P(H^*)$ into functions $b\in \pSub{\one}(H)$, such that $\gauge^\Sub_H(q_{M,\alpha})$ is very low. 
More precisely, we will apply the following principle to make use of \Cref{thm:flow.engine}. To simplfy the presentation, the following  considers only hypergraphs $H$ where $P(H^*)$ has at least one edge (i.e., there are two hyperedges that have non-empty intersection).

\begin{definition}
For a hypergraph $H$ and $M:=P(H^*)$, define its \emph{(submodular) gauge-routing ratio}
\[
\cgamp(H):=
\sup_{\alpha\neq 0}
\frac{\gauge_H^\Sub(q_{M,\alpha})}
{\max_{e\in E(H)}\alpha(\Inc_M(e))}
\]
where the supremum ranges over all nonzero edge profiles $\alpha : E(M) \to \rpos$.
\end{definition}

The submodular gauge-routing ratio simply measures how efficiently incident load in the line graph can be turned
into submodular gauge. Combined with \Cref{thm:flow.engine}, it immediately
yields lower bounds on $\subw(H)$. 

\begin{lemma}
\label{lem:load.to.width}
There exists a universal constant $C>0$ such that, for every hypergraph $H$
with $\ghw(H)\ge 9$,
\[
\subw(H)\ge
\frac{C}{\cgamp(H)}\cdot
\frac{\ghw(H)}{\log \ghw(H)}.
\]
\end{lemma}

\begin{proof}
By \Cref{prop:support-tw-controls-ghw}, $\ghw(H)\le \tw(P(H^*))+1$, and therefore
$\tw(P(H^*))\ge \ghw(H)-1\ge 8$.
Hence \Cref{thm:flow.engine} provides an edge profile
$\alpha\in \rpos^{E(P(H^*))}$ such that
\[
q:=q_{P(H^*),\alpha}\in \Sym(H),
\qquad
\disp_H(q)\ge \frac49,
\]
and
\[
\tau:=\max_{e\in E(H)}\alpha(\Inc_{P(H^*)}(e))
\le
C_0\,\frac{\log \tw(P(H^*))}{\tw(P(H^*))}
\]
for some universal constant $C_0>0$.

By definition of $\cgamp(H)$,
\[
\gauge_H^\Sub(q)\le \tau\,\cgamp(H)
\le
C_0\,\cgamp(H)\,\frac{\log \tw(P(H^*))}{\tw(P(H^*))}.
\]
Therefore, by \Cref{thm:variational},
\[
\subwlift(H)
\ge
\frac{\disp_H(q)}{\gauge_H^\Sub(q)}
\ge
\frac{4}{9\,C_0\,\cgamp(H)}
\cdot
\frac{\tw(P(H^*))}{\log \tw(P(H^*))}.
\]
Since $\tw(P(H^*))\ge \ghw(H)-1$, this implies
\[
\subwlift(H)\ge
\frac{C_1}{\cgamp(H)}
\cdot
\frac{\ghw(H)}{\log \ghw(H)}
\]
for some universal constant $C_1>0$.
Finally, \Cref{cor:lifted-bw} gives $\subw(H)\ge \subwlift(H)$.
\end{proof}

On a technical level, the way we will bound $\cgamp(H)$ is by constructing a $b \in \pSub\one(H)$ from every $\alpha$. This may raise the question why this technically differs from constructing $b$ right away. There are multiple subtle but important reasons. The developments from \Cref{sec:support-inputs} construct technically very useful $q\in\Sym(H)$ for all hypergraphs. We can use this to bound $\subw$ in terms of $\ghw$ in those cases where this $q$ has small gauge (in particular, small $\cgamp(H)$). Using the tools from \Cref{sec:support-inputs} directly to construct a $b \in \pSub\one(H)$ would require applying them directly under some constraints to the structure of $H$, leading to much higher conceptual complexity. That is, the variational characterisation allows us to separate principles construction of $q$, from its application to bounding $\subw$, leading to much simpler and arguably more natural arguments overall. Furthermore, the cut perspective made it much easier to construct our ``canonical'' $q \in \Sym(H)$, that guarantees a universal lower bound on $\disp_H$. However, this implies high submodular width through \Cref{thm:lift.to.tau} only when $q$ is a boundary lift of $b \in \pSub\one$, a much stricter property than what we use to bound the gauges here.

We now discuss two natural but general structural properties that imply bounded $\cgamp$. The proofs are simple but instructive. We defer them to Appendix~\ref{app:applications}.
First, we consider the role of what we call the \emph{edge excess} -- the sum of degrees greater than 2 in an edge. That is, degree 2 vertices do not count towards the excess, and the total degree of the rest of the vertices must be low. In a database setting, where degree represents how many joins an attribute is involved in, low edge excess already covers most natural cases. For hypergraphs with maximum degree 2, the edge excess is 0.

\begin{definition}
For a hypergraph $H$, define its \emph{edge excess} as
\[
\eex(H):=
\max_{e\in E(H)}
\sum_{v\in e}\max\{\deg_H(v)-2,0\}.
\]
\end{definition}

\begin{proposition}
\label{bounded-excess-ratio}
For every hypergraph $H$, it holds that $\cgamp(H) \le \eex(H)+1$.
\end{proposition}

When we are only interested in the qualitative goal of observing when unbounded $\ghw$ implies unbounded $\subw$ on a class of hypergraphs, constant edge excess is not essential. Indeed, if a class $\mathcal H$ satisfies
$\eex(H)\in O((\log \ghw(H))^c)$
for some fixed $c$ on all $H\in \mathcal H$,
then \Cref{bounded-excess-ratio} yields
\[
\subw(H)\in \Omega\!\left(\frac{\ghw(H)}{(\log \ghw(H))^{c+1}}\right).
\]
In particular, every such class with unbounded $\ghw$ also has unbounded $\subw$. We illustrate an example of such a class in Appendix~\ref{app:excess}.

The second constraint that generalises maximum degree 2 is of a very different flavour. Rather than controlling the total amount of high-degree overlap, it merely requires that each pairwise intersection contain at least one degree-$2$ witness. In particular, it places no direct restriction on the rest of the hypergraph, and yet it is already strong enough to produce a strong asymptotic lower bound on $\subw$ in terms of $\ghw$.

\begin{definition}
We say that $H$ has the \emph{private intersections  property} if, for
every distinct $x,y\in E(H)$ with $x\cap y\neq \emptyset$, there exists
$v_{xy}\in x\cap y$ with $\deg_H(v_{xy})=2$.
\end{definition}

\begin{proposition}
\label{prop:exclusive-witness-ratio}
Assume that $H$ has the private intersection property. Then $\cgamp(H) \le 1$.
\end{proposition}

The private intersection property is in its own sense much more permissive than a global degree bound. A hypergraph may satisfy it while still having vertices of arbitrarily large degree, large pairwise intersections, or highly complicated higher-order overlap (e.g., high VC-dimension). The point is that our realisation argument needs only one degree-$2$ witness for each line-graph edge $xy$, and once such a witness exists, all remaining intersection structure can be ignored. We consider this rather surprising, and given the later complexity consequences it may be worthwhile to study  direct algorithmic consequences of this property further in a database context.

That is, by combining these bounds with \Cref{lem:load.to.width} we can observe
\[
\subw(H)
\in \Omega\left(\frac{1}{\eex(H)+1}\cdot \frac{\ghw(H)}{\log \ghw(H)}\right)
\]
in general. And $\subw(H) \in \Omega(\ghw(H)/\log \ghw(H))$ for $H$ with private intersections.

\subsection{A General Fractional Witness Principle}
\label{sec:kappa}
Abstracting one level, the two previous applications are instances of the same mechanism. For each line-graph edge $xy\in E(P(H^*))$, one must assign one unit of witness mass to vertices of the intersection $x\cap y$. The core question then is how much total charge this induces on each hyperedge of $H$. In the private-intersection case this assignment is integral and supported on a single degree $2$ vertex, whereas in the bounded-excess case it is distributed over larger intersections. Allowing fractional assignments isolates the common optimisation problem and turns the resulting cost into a linear, hence modular, functional of the line graph.

\begin{definition}
Let $H$ be a hypergraph, and let $M:=P(H^*)$. An \emph{intersection capacity allocation} on $H$ is a family $p$ of all $p_{xy,v} \in \rpos$ where $xy \in E(M)$ and $v \in x \cap y$ such that
\[
\sum_{v\in x\cap y} p_{xy,v}=1
\qquad\forall xy\in E(M).
\]
For $e\in E(H)$ and $xy\in E(M)$, define
\[
w_e^p(xy):=\sum_{v\in e\cap x\cap y} p_{xy,v}.
\]
\end{definition}

For an intersection capacity allocation we now want to measure the maximum fractional load it can induce on any individual hyperedge.
We write 
\[
\kappa_e(p):=
\max\Bigl\{
\sum_{xy\in E(M)} w_e^p(xy)\beta_{xy}
\mid \ \beta\in \rpos^{E(M)},\
\beta(\Inc_M(h))\le 1\ \forall h\in E(H)
\Bigr\},
\]
and $\kappa(p):=\max_{e\in E(H)}\kappa_e(p)$.
Finally, define $\kappa(H):=\min_p \kappa(p)$,
where the minimum ranges over all intersection capacity allocations $p$ on $H$.

\begin{lemma}
\label{prop:fractional-witness-realisation}
Let $H$ be a hypergraph, let $M:=P(H^*)$, let $\alpha\in \rpos^{E(M)}$, and let
$p$ be an intersection capacity allocation on $H$. Then
$\gauge_H^\Sub(q_{M,\alpha})
\le
\kappa(p)\max_{e\in E(H)}\alpha(\Inc_M(e))$.
In particular,
$\cgamp(H)\le \kappa(H)$.
\end{lemma}

\begin{proof}
Set
\(
\tau:=\max_{h\in E(H)}\alpha(\Inc_M(h)),
\)
and define
\[
\mu(v):=\sum_{xy\in E(M):\,v\in x\cap y}\alpha_{xy}p_{xy,v},
\qquad
b(X):=\sum_{v\in X}\mu(v).
\]
Then $b$ is modular, hence also submodular.

We first show that $q_{M,\alpha}\le \lambda_b$.
Fix $F\subseteq E(H)$. If $xy\in \cut{M}(F)$, then $x$ and $y$ lie on opposite
sides of the cut, so every vertex of $x\cap y$ belongs to $\partial_H(F)$.
Since $\sum_{v\in x\cap y}p_{xy,v}=1$, the full mass $\alpha_{xy}$ contributes
to $b(\partial_H(F))$. Summing over all $xy\in\cut{M}(F)$ gives
\(
q_{M,\alpha}(F)\le b(\partial_H(F))=\lambda_b(F).
\)

It remains to bound $b$ on hyperedges.
Fix $e\in E(H)$. Since $e$ is the hyperedge whose load we are measuring, we use
$h$ for the generic hyperedge in the constraints defining $\kappa_e(p)$.
By rearranging the sums,
\[
b(e)=\sum_{xy\in E(M)}\alpha_{xy}w_e^p(xy).
\]

If $\tau=0$, then $\alpha=0$, hence $q_{M,\alpha}=0$, and there is nothing to
prove. So assume $\tau>0$, and set $\beta:=\alpha/\tau$. Then
$\beta(\Inc_M(h))\le 1$ for every $h\in E(H)$.
Therefore, by the definition of $\kappa_e(p)$,
\[
b(e)
=
\tau\sum_{xy\in E(M)}w_e^p(xy)\beta_{xy}
\le
\tau\,\kappa_e(p)
\le
\tau\,\kappa(p).
\]
Since $e$ was arbitrary, we have $b(h)\le \tau\kappa(p)$ for every
$h\in E(H)$.

If $\tau\kappa(p)=0$, then $b=0$, and hence again $q_{M,\alpha}=0$.
Otherwise define
\[
\widehat b:=\frac{1}{\tau\kappa(p)}\,b.
\]
Then $\widehat b\in \pMod{\one}(H)\subseteq \pSub{\one}(H)$, which implies
\(
q_{M,\alpha}\le \lambda_b=\tau\kappa(p)\,\lambda_{\widehat b},
\)
and therefore also
\(
\gauge_H^\Sub(q_{M,\alpha})\le \tau\kappa(p).
\)
Since $p$ was arbitrary, taking the infimum over all intersection capacity
allocations yields $\cgamp(H)\le \kappa(H)$.
\end{proof}

\begin{corollary}
\label{cor:kappa-subw}
There exists a universal constant $C>0$ such that, for every hypergraph $H$
with $\ghw(H)\ge 9$,
\[
\subw(H)\ge
\frac{C}{\kappa(H)}\cdot
\frac{\ghw(H)}{\log \ghw(H)}.
\]
\end{corollary}

\begin{proof}
Immediate from \Cref{lem:load.to.width} and
\Cref{prop:fractional-witness-realisation}.
\end{proof}

Note that previous applications can be recovered by choosing specific intersection capacity allocations analogous to the arguments in~\Cref{app:applications}. Under the private intersection property one has
$\kappa(H)\le 1$, while bounded edge excess gives
$\kappa(H)\le \eex(H)+1$. 

There is a particularly interesting observation to make here that should be explored further. The direct argument for the $\eex(H)+1$  bound proves the gauge is small relative to some coverage function. In the proof a coverage function is the natural choice and it is not clear how to relax this to show the gauge relative to a modular function. In contrast, \Cref{prop:fractional-witness-realisation} in fact always shows that the gauge $\gauge^\Mod_{H,\one}$ (the gauge relative to $K^\Mod_\one$) is bounded in terms of $\kappa$, as the bound comes from a modular $b$. That is, this more general principle implicitly translates our realisability argument for edge excess from one based on submodular functions, to one based on modular functions.

Concretely, this implies the same lower bounds already for adaptive width. But on a higher level this hints at the opportunity for deeper insight into the connections between adaptive and submodular width. 
New insight into this connection might be gained through studying the relationship of submodular and modular gauges for the realisable body.

\section{Complexity Consequences}
\label{sec:consequences}
We assume the reader to be familiar with \emph{conjunctive queries} (CQs). We use the notation and terminology of~\cite{pdm}. In particular, for CQ $q$ and database $D$, we write $q(D)$ for the set of answers of $q$ over $D$.
We write $\cqprob(\cH)$ for the Boolean conjunctive query evaluation problem over hypergraph class $\cH$. Namely, given a CQ $q$ with $H(q) \in \cH$ and a database $D$, is $q(D)\neq \emptyset$?
We also write $\pcq(\cH)$ for the typical parameterisation of the problem studied in the context of parameterised complexity, with the same inputs and output as $\cqprob(\cH)$, and parameter $q$. That is, we say $\pcq(\cH) \in \mathsf{FPT}$ if there is an $f(q)\,\mathsf{poly}(D)$ time algorithm for the problem, where $f$ is computable. 

\begin{proposition}[\citet{DBLP:journals/ejc/AdlerGG07,DBLP:journals/jcss/GottlobLS02}]
  \label{prop:hwtrac}
  Let $\cH$ be a class of hypergraphs with constantly bounded $\ghw$. Then $\cqprob(\cH) \in \mathsf{PTIME}$.
\end{proposition}

\begin{proposition}[\citet{Marx13}]
\label{prop:subwfpt}
Let $\cH$ be a recursively enumerable class of hypergraphs. Then, assuming ETH,
$\pcq(\cH) \in \mathsf{FPT}$  if and only if $\cH$ has constantly bounded $\subw$.
\end{proposition}

We extend hypergraph parameters to classes of hypergraphs in the natural way. That is, for a class $\cH$ we write $f(\cH)$ to mean $\sup_{H\in \cH} f(H)$. Moreover, we write $f(\cH) < \infty$ to state that there exists a constant $c\in \rpos$ such that $f(\cH) \le c$.

We can now formulate our main consequence in terms of computational complexity. Analogous to the bounded rank dichotomy for CQ evaluation by Grohe~\cite{DBLP:journals/jacm/Grohe07} and the aforementioned degree 2 result~\cite{LanzingerPODS22}, we observe a dichotomy where $\cqprob(\cH)$ is  either solvable in polynomial time, or $\pcq(\cH)$ is not fixed-parameter tractable.

\begin{theorem}\label{thm:complexity}
  Let $\cH$ be a recursively enumerable class of hypergraphs with bounded submodular gauge-routing ratio, i.e.,  $\cgamp(\cH)<\infty$. 
Assuming ETH, the following are equivalent:
\begin{enumerate}[label=(\roman*)]
\item $\ghw(\cH) < \infty$.
\item $\subw(\cH)< \infty$.
\item $\subwlift(\cH)< \infty$.
\item $\pcq(\cH) \in \mathsf{FPT}$.
\item $\cqprob(\cH) \in \mathsf{PTIME}$.
\end{enumerate}
\end{theorem}
\begin{proof}
By \Cref{cor:lifted-bw} and \Cref{prop:subwfpt}, (ii) $\iff$ (iii) $\iff $ (iv). Assuming bounded $\cgamp$, \Cref{ghw.subw} and \Cref{lem:load.to.width}, also (i) is equivalent to those statements.
To close the final loop, observe that (i) $\implies$ (v) by \Cref{prop:hwtrac}, and (v) $\implies$ (iv) trivially. 
\end{proof}

\Cref{thm:complexity} should be viewed as a meta-theorem.
The previous sections identified several natural sufficient conditions for bounded
gauge-routing ratio, including bounded edge excess, exclusive intersection structure,
and the more general routing principles from \Cref{sec:kappa}.
Taken together, these results clarify how tightly $\subw$ and $\ghw$ are linked on a
broad family of hypergraph classes.

So far we have focused our investigations on generalised hypertree width $\ghw$, i.e., $\tau_b$ where $b$ is the (integral) edge cover number. Relaxing this to the case where $b$ is the \emph{fractional} edge cover number defines the important notion of \emph{fractional hypertree width} $\fhw$~\cite{fhw}, that is known to induce tractable classes for $\cqprob$ beyond bounded $\ghw$.
It is well-known that $\subw(H) \le \fhw(H) \le \ghw(H)$. Hence $\cgamp(\cH)<\infty$ also implies that $\ghw(\cH)<\infty$  iff $\fhw(\cH)<\infty$.  That is, for the classes investigated here the notions collapse in terms of boundedness. Lower bounds based particularly on $\fhw$ are a natural opportunity for future work. We discuss this matter in  detail in~\Cref{sec:conclusion}.

$\cqprob$ and $\pcq$ are also studied in terms of restrictions by query classes rather than hypergraph classes. This introduces one extra step of complications, as queries might have structurally simple equivalent \emph{cores}.
Our results extend naturally to the setting of query classes by results of Chen et al.~\cite{DBLP:conf/ijcai/ChenGLP20}; see also \cite[Section~4.3]{LanzingerPODS22} for details.

\section{Conclusion \& Outlook}
\label{sec:conclusion}
\subsection{Conclusion}
We propose a geometric reframing of submodular width. Instead of reasoning directly about tree decompositions for each admissible submodular function, we pass to symmetric edge-set functions and study their position inside the realisable submodular body $K^\Sub_{\one}(H)$. Our variational characterisation recasts submodular width as a geometric question about the hypergraph itself. Instead of treating each admissible submodular function separately, one studies the shape of the realisable submodular body and the extent to which symmetric cut functions can be accommodated within it. In this way, a problem originally phrased in terms of decompositions and adversarial submodular functions is recast as a finite-dimensional convex optimisation problem.

This geometric reformulation presents us with a new task: one needs a systematic way of producing symmetric cut functions that are sufficiently well distributed to witness large width. We develop such a method by applying classical multicommodity-flow results on the line graph of the studied hypergraph. This yields symmetric cut functions with  controlled incident load at each hyperedge. In a further step we show that small incident load also forces small gauge in the realisable body. Our applications show that this mechanism is rather robust and we are able to bound $\subw$ in terms of $\ghw$ under much looser conditions than previously known.

Finally, these structural results have concrete consequences for the complexity of conjunctive query evaluation. In particular, on every recursively enumerable class with bounded gauge-routing ratio, polynomial-time solvability and fixed-parameter tractability coincide under ETH.

\subsection{Directions for Future Research}
\label{sec:future}
The geometric framework developed in this paper recasts the study of submodular width as a collection of questions about a finite-dimensional convex body. This change in perspective opens up a broader landscape of problems, only a small part of which is explored here. On the more immediate side, the results of \Cref{sec:lift} and \Cref{sec:geometry}  extend naturally to other width measures in the spirit of adaptive or submodular width, for instance those defined using entropic or polymatroid functions. It is also natural to ask how this geometric viewpoint interacts with data-constraint-aware forms of conjunctive query evaluation, such as those studied in~\cite{DBLP:conf/pods/Khamis0S17,panda25}. Among these and many other possibilities, we highlight the following directions for further work.

\paragraph{Computing Submodular Width}
The geometric viewpoint suggests a new algorithmic approach to computing submodular width of a hypergraph. Historically, deciding or computing submodular width, or finding certificates for low width has been highly challenging. Recently, there has been some progress as Abo Khamis et al.~\cite{DBLP:journals/pacmmod/KhamisHS25} presented a general algorithm for computing the more general $\omega$-$
\subw$ through finitely many linear programs, but it remains unclear whether this is practical.
Rather than searching directly over tree decompositions and adversarial submodular witnesses, one could phrase the problem in terms of optimisation over the realisable body and its gauge. This shifts the difficulty into a finite-dimensional convex setting, where one may hope for separation
procedures, approximation algorithms, or certificate systems that are more approachable than the original definition.

From this perspective, the point is not only to ask whether submodular width can be computed exactly. More broadly, one can ask whether the variational description and the associated antiblocker duality can be turned into effective optimisation tools. The primal body describes which symmetric cut functions can be realised at a given scale, while the antiblocker offers a dual language for certifying when such realisation is impossible. If these two sides can be made algorithmic, then the geometric viewpoint could provide a practical way of estimating, certifying, and investigating submodular width.

\paragraph{Technical Approaches Beyond \Cref{sec:support-inputs}}

The route through the line graph $P(H^*)$ is only one way of building on the variational characterisation. It is effective because multicommodity flow on the line graph produces balanced cut weightings with small incident load,
  but this is only one technical option towards identifying relevant cut weightings. 
In this sense, the line graph should be viewed as a particularly convenient first interface between routing and overlap, not as the unique setting in which the geometric method can operate.

We believe there is room for several alternative approaches. One may try to construct
good $q \in \Sym(H)$ directly from separators in the hypergraph or in the incidence
graph, from metric or spectral relaxations, or from higher-order overlap
structures that are not visible in the pairwise line graph alone. Another
natural direction is to look for mechanisms for constructing symmetric cut functions that already
reflect fractional cover structure to relate submodular width also more directly to $\fhw$ rather than only to $\ghw$.

\paragraph{Polyhedral Structure of the Realisable Submodular Body}
For a fixed hypergraph $H$, the set $K^\Sub_{\one}(H)\subseteq \rpos^{2^{E(H)}}$
is a finite-dimensional polytope.

 The current
framework shows that lower bounds on submodular width are governed by
membership and scaling questions for a polyhedral body, but the geometry is
still encoded implicitly through submodular witnesses. An explicit description
in cut coordinates could turn existential representation into concrete
optimisation, provide recognisable certificates of membership and
non-membership, and make it possible to study approximation algorithms for the
gauge in a more direct way.

From the hypergraph point of view, the central issue is to relate the
intersection structure of $H$ to the polyhedral structure of
$K^\Sub_{\one}(H)$. Which overlap configurations force valid inequalities, and
when do these become facet-defining? Which vertices of
$K^\Sub_{\one}(H)$ can be described combinatorially? How does the polytope
transform under natural operations on $H$, such as gluing along separators,
subdividing intersections, or adjoining vertices with controlled incidence?
Clarifying these questions would make the geometry more explicit, and could
help connect it more closely to the algorithmic theory of conjunctive query
evaluation.

\paragraph{Sparse Representations of Near-Optimal Symmetric Cut Functions}

A related geometric question concerns the complexity of representing the symmetric cut functions that are relevant to submodular width. Since $K^\Sub_{\one}(H)\cap \Sym(H)$ is a polytope, every element of this set can be written as a convex combination of finitely many extreme points. The ambient-dimensional bound  on the complexity of this combination is immediate, but is likely to be far from sharp. This raises the question of whether symmetric cut functions that are near-optimal for $\subwlift(H)$ always admit \emph{sparse} representations.

In our setting, such a result would have a concrete meaning. It would say that lower bounds relevant to submodular width can be witnessed by combining only a small number of basic realisable cut functions, rather than by a highly distributed convex combination. This would be particularly interesting if the required support size could be bounded in terms of structural parameters of $H$, such as rank, intersection behaviour, or VC-dimension.

This question also has a natural algorithmic interpretation. If near-optimal witnesses always admitted succinct decompositions, then one could hope to search for strong lower-bound certificates in a much smaller space, or to design hybrid evaluation procedures in which the more expensive submodular-width machinery is only invoked on parts of the query that genuinely require it.

\paragraph{Fractional Witness Assignments as Structural Parameters}

The parameter $\kappa(H)$ from \Cref{sec:kappa} already subsumes the two main applications of this paper: private intersections give $\kappa(H)\le 1$, while bounded edge excess gives $\kappa(H)\le \eex(H)+1$. In that sense, $\kappa(H)$ is not merely an auxiliary quantity, but a general
structural parameter extracted from the realisability argument itself.

What is still missing is a theory explaining when $\kappa(H)$ is small and how sharply it captures the gauge-routing ratio. For fixed $p$ and $e$, the quantity $\kappa_e(p)$ is a linear-programming optimum, and LP duality gives
it the flavour of a weighted fractional cover problem on the line graph $P(H^*)$. This suggests that $\kappa(H)$ should admit a more transparent combinatorial interpretation than is currently visible from the definition.

\bibliographystyle{ACM-Reference-Format}
\bibliography{references}


\begin{thebibliography}{28}


\ifx \showCODEN    \undefined \def \showCODEN     #1{\unskip}     \fi
\ifx \showISBNx    \undefined \def \showISBNx     #1{\unskip}     \fi
\ifx \showISBNxiii \undefined \def \showISBNxiii  #1{\unskip}     \fi
\ifx \showISSN     \undefined \def \showISSN      #1{\unskip}     \fi
\ifx \showLCCN     \undefined \def \showLCCN      #1{\unskip}     \fi
\ifx \shownote     \undefined \def \shownote      #1{#1}          \fi
\ifx \showarticletitle \undefined \def \showarticletitle #1{#1}   \fi
\ifx \showURL      \undefined \def \showURL       {\relax}        \fi
\providecommand\bibfield[2]{#2}
\providecommand\bibinfo[2]{#2}
\providecommand\natexlab[1]{#1}
\providecommand\showeprint[2][]{arXiv:#2}

\bibitem[Adler et~al\mbox{.}(2007)]%
        {DBLP:journals/ejc/AdlerGG07}
\bibfield{author}{\bibinfo{person}{Isolde Adler}, \bibinfo{person}{Georg
  Gottlob}, {and} \bibinfo{person}{Martin Grohe}.}
  \bibinfo{year}{2007}\natexlab{}.
\newblock \showarticletitle{Hypertree width and related hypergraph invariants}.
\newblock \bibinfo{journal}{\emph{Eur. J. Comb.}} \bibinfo{volume}{28},
  \bibinfo{number}{8} (\bibinfo{year}{2007}), \bibinfo{pages}{2167--2181}.
\newblock
\href{https://doi.org/10.1016/j.ejc.2007.04.013}{doi:\nolinkurl{10.1016/j.ejc.2007.04.013}}


\bibitem[Arenas et~al\mbox{.}(2022)]%
        {pdm}
\bibfield{author}{\bibinfo{person}{Marcelo Arenas}, \bibinfo{person}{Pablo
  Barcel\'o}, \bibinfo{person}{Leonid Libkin}, \bibinfo{person}{Wim Martens},
  {and} \bibinfo{person}{Andreas Pieris}.} \bibinfo{year}{2022}\natexlab{}.
\newblock \bibinfo{booktitle}{\emph{Database Theory}}.
\newblock \bibinfo{publisher}{Open source at
  \url{https://github.com/pdm-book/community}}.
\newblock


\bibitem[Bulatov(2017)]%
        {Bulatov17}
\bibfield{author}{\bibinfo{person}{Andrei~A. Bulatov}.}
  \bibinfo{year}{2017}\natexlab{}.
\newblock \showarticletitle{A Dichotomy Theorem for Nonuniform CSPs}. In
  \bibinfo{booktitle}{\emph{58th {IEEE} Annual Symposium on Foundations of
  Computer Science, {FOCS} 2017, Berkeley, CA, USA, October 15-17, 2017}},
  \bibfield{editor}{\bibinfo{person}{Chris Umans}} (Ed.).
  \bibinfo{publisher}{{IEEE} Computer Society}, \bibinfo{pages}{319--330}.
\newblock
\href{https://doi.org/10.1109/FOCS.2017.37}{doi:\nolinkurl{10.1109/FOCS.2017.37}}


\bibitem[Chekuri and Chuzhoy(2013)]%
        {ChekuriChuzhoy13}
\bibfield{author}{\bibinfo{person}{Chandra Chekuri} {and}
  \bibinfo{person}{Julia Chuzhoy}.} \bibinfo{year}{2013}\natexlab{}.
\newblock \showarticletitle{Large-treewidth graph decompositions and
  applications}. In \bibinfo{booktitle}{\emph{Proceedings of the 45th Annual
  {ACM} Symposium on Theory of Computing Conference, STOC 2013}}.
  \bibinfo{publisher}{{ACM}}, \bibinfo{pages}{291--300}.
\newblock
\href{https://doi.org/10.1145/2488608.2488645}{doi:\nolinkurl{10.1145/2488608.2488645}}


\bibitem[Chekuri et~al\mbox{.}(2005)]%
        {ChekuriKhannaShepherd05}
\bibfield{author}{\bibinfo{person}{Chandra Chekuri}, \bibinfo{person}{Sanjeev
  Khanna}, {and} \bibinfo{person}{F.~Bruce Shepherd}.}
  \bibinfo{year}{2005}\natexlab{}.
\newblock \showarticletitle{Multicommodity flow, well-linked terminals, and
  routing problems}. In \bibinfo{booktitle}{\emph{Proceedings of the 37th
  Annual {ACM} Symposium on Theory of Computin, STOC 2005}}.
  \bibinfo{publisher}{{ACM}}, \bibinfo{pages}{183--192}.
\newblock
\href{https://doi.org/10.1145/1060590.1060618}{doi:\nolinkurl{10.1145/1060590.1060618}}


\bibitem[Chen et~al\mbox{.}(2020)]%
        {DBLP:conf/ijcai/ChenGLP20}
\bibfield{author}{\bibinfo{person}{Hubie Chen}, \bibinfo{person}{Georg
  Gottlob}, \bibinfo{person}{Matthias Lanzinger}, {and}
  \bibinfo{person}{Reinhard Pichler}.} \bibinfo{year}{2020}\natexlab{}.
\newblock \showarticletitle{Semantic Width and the Fixed-Parameter Tractability
  of Constraint Satisfaction Problems}. In
  \bibinfo{booktitle}{\emph{Proceedings of the Twenty-Ninth International Joint
  Conference on Artificial Intelligence, {IJCAI} 2020}}.
  \bibinfo{publisher}{ijcai.org}, \bibinfo{pages}{1726--1733}.
\newblock
\href{https://doi.org/10.24963/IJCAI.2020/239}{doi:\nolinkurl{10.24963/IJCAI.2020/239}}


\bibitem[Csisz{\'{a}}r et~al\mbox{.}(1990)]%
        {DBLP:journals/combinatorica/CsiszarKLMS90}
\bibfield{author}{\bibinfo{person}{Imre Csisz{\'{a}}r},
  \bibinfo{person}{J{\'{a}}nos K{\"{o}}rner},
  \bibinfo{person}{L{\'{a}}szl{\'{o}} Lov{\'{a}}sz}, \bibinfo{person}{Katalin
  Marton}, {and} \bibinfo{person}{G{\'{a}}bor Simonyi}.}
  \bibinfo{year}{1990}\natexlab{}.
\newblock \showarticletitle{Entropy splitting for antiblocking corners and
  perfect graphs}.
\newblock \bibinfo{journal}{\emph{Comb.}} \bibinfo{volume}{10},
  \bibinfo{number}{1} (\bibinfo{year}{1990}), \bibinfo{pages}{27--40}.
\newblock
\href{https://doi.org/10.1007/BF02122693}{doi:\nolinkurl{10.1007/BF02122693}}


\bibitem[Feige et~al\mbox{.}(2008)]%
        {FeigeHajiaghayiLee08}
\bibfield{author}{\bibinfo{person}{Uriel Feige}, \bibinfo{person}{MohammadTaghi
  Hajiaghayi}, {and} \bibinfo{person}{James~R. Lee}.}
  \bibinfo{year}{2008}\natexlab{}.
\newblock \showarticletitle{Improved Approximation Algorithms for Minimum
  Weight Vertex Separators}.
\newblock \bibinfo{journal}{\emph{{SIAM} J. Comput.}} \bibinfo{volume}{38},
  \bibinfo{number}{2} (\bibinfo{year}{2008}), \bibinfo{pages}{629--657}.
\newblock
\href{https://doi.org/10.1137/05064299X}{doi:\nolinkurl{10.1137/05064299X}}


\bibitem[Fulkerson(1972)]%
        {fulkerson1972anti}
\bibfield{author}{\bibinfo{person}{Delbert~R Fulkerson}.}
  \bibinfo{year}{1972}\natexlab{}.
\newblock \showarticletitle{Anti-blocking polyhedra}.
\newblock \bibinfo{journal}{\emph{Journal of Combinatorial Theory, Series B}}
  \bibinfo{volume}{12}, \bibinfo{number}{1} (\bibinfo{year}{1972}),
  \bibinfo{pages}{50--71}.
\newblock


\bibitem[Gottlob et~al\mbox{.}(2021)]%
        {DBLP:journals/jacm/GottlobLPR21}
\bibfield{author}{\bibinfo{person}{Georg Gottlob}, \bibinfo{person}{Matthias
  Lanzinger}, \bibinfo{person}{Reinhard Pichler}, {and} \bibinfo{person}{Igor
  Razgon}.} \bibinfo{year}{2021}\natexlab{}.
\newblock \showarticletitle{Complexity Analysis of Generalized and Fractional
  Hypertree Decompositions}.
\newblock \bibinfo{journal}{\emph{J. {ACM}}} \bibinfo{volume}{68},
  \bibinfo{number}{5} (\bibinfo{year}{2021}), \bibinfo{pages}{38:1--38:50}.
\newblock
\href{https://doi.org/10.1145/3457374}{doi:\nolinkurl{10.1145/3457374}}


\bibitem[Gottlob et~al\mbox{.}(2002)]%
        {DBLP:journals/jcss/GottlobLS02}
\bibfield{author}{\bibinfo{person}{Georg Gottlob}, \bibinfo{person}{Nicola
  Leone}, {and} \bibinfo{person}{Francesco Scarcello}.}
  \bibinfo{year}{2002}\natexlab{}.
\newblock \showarticletitle{Hypertree Decompositions and Tractable Queries}.
\newblock \bibinfo{journal}{\emph{J. Comput. Syst. Sci.}} \bibinfo{volume}{64},
  \bibinfo{number}{3} (\bibinfo{year}{2002}), \bibinfo{pages}{579--627}.
\newblock
\href{https://doi.org/10.1006/jcss.2001.1809}{doi:\nolinkurl{10.1006/jcss.2001.1809}}


\bibitem[Gottlob et~al\mbox{.}(2009)]%
        {DBLP:journals/jacm/GottlobMS09}
\bibfield{author}{\bibinfo{person}{Georg Gottlob},
  \bibinfo{person}{Zolt{\'{a}}n Mikl{\'{o}}s}, {and} \bibinfo{person}{Thomas
  Schwentick}.} \bibinfo{year}{2009}\natexlab{}.
\newblock \showarticletitle{Generalized hypertree decompositions: NP-hardness
  and tractable variants}.
\newblock \bibinfo{journal}{\emph{J. {ACM}}} \bibinfo{volume}{56},
  \bibinfo{number}{6} (\bibinfo{year}{2009}), \bibinfo{pages}{30:1--30:32}.
\newblock
\href{https://doi.org/10.1145/1568318.1568320}{doi:\nolinkurl{10.1145/1568318.1568320}}


\bibitem[Grohe(2007)]%
        {DBLP:journals/jacm/Grohe07}
\bibfield{author}{\bibinfo{person}{Martin Grohe}.}
  \bibinfo{year}{2007}\natexlab{}.
\newblock \showarticletitle{The complexity of homomorphism and constraint
  satisfaction problems seen from the other side}.
\newblock \bibinfo{journal}{\emph{J. {ACM}}} \bibinfo{volume}{54},
  \bibinfo{number}{1} (\bibinfo{year}{2007}), \bibinfo{pages}{1:1--1:24}.
\newblock
\href{https://doi.org/10.1145/1206035.1206036}{doi:\nolinkurl{10.1145/1206035.1206036}}


\bibitem[Grohe(2016)]%
        {Grohe16}
\bibfield{author}{\bibinfo{person}{Martin Grohe}.}
  \bibinfo{year}{2016}\natexlab{}.
\newblock \showarticletitle{Tangled up in Blue {(A} Survey on Connectivity,
  Decompositions, and Tangles)}.
\newblock \bibinfo{journal}{\emph{CoRR}}  \bibinfo{volume}{abs/1605.06704}
  (\bibinfo{year}{2016}).
\newblock
\showeprint[arXiv]{1605.06704}
\urldef\tempurl%
\url{http://arxiv.org/abs/1605.06704}
\showURL{%
\tempurl}


\bibitem[Grohe and Marx(2006)]%
        {fhw}
\bibfield{author}{\bibinfo{person}{Martin Grohe} {and}
  \bibinfo{person}{D{\'{a}}niel Marx}.} \bibinfo{year}{2006}\natexlab{}.
\newblock \showarticletitle{Constraint solving via fractional edge covers}. In
  \bibinfo{booktitle}{\emph{Proceedings of the Seventeenth Annual {ACM-SIAM}
  Symposium on Discrete Algorithms, {SODA} 2006}}. \bibinfo{publisher}{{ACM}
  Press}, \bibinfo{pages}{289--298}.
\newblock
\urldef\tempurl%
\url{http://dl.acm.org/citation.cfm?id=1109557.1109590}
\showURL{%
\tempurl}


\bibitem[Khamis and Chen(2026)]%
        {jaguar}
\bibfield{author}{\bibinfo{person}{Mahmoud~Abo Khamis} {and}
  \bibinfo{person}{Hubie Chen}.} \bibinfo{year}{2026}\natexlab{}.
\newblock \showarticletitle{Jaguar: {A} Primal Algorithm for Conjunctive Query
  Evaluation in Submodular-Width Time}.
\newblock \bibinfo{journal}{\emph{CoRR}}  \bibinfo{volume}{abs/2603.13624}
  (\bibinfo{year}{2026}).
\newblock
\href{https://doi.org/10.48550/ARXIV.2603.13624}{doi:\nolinkurl{10.48550/ARXIV.2603.13624}}
\showeprint[arXiv]{2603.13624}


\bibitem[Khamis et~al\mbox{.}(2025a)]%
        {DBLP:journals/pacmmod/KhamisHS25}
\bibfield{author}{\bibinfo{person}{Mahmoud~Abo Khamis}, \bibinfo{person}{Xiao
  Hu}, {and} \bibinfo{person}{Dan Suciu}.} \bibinfo{year}{2025}\natexlab{a}.
\newblock \showarticletitle{Fast Matrix Multiplication meets the Submodular
  Width}.
\newblock \bibinfo{journal}{\emph{Proc. {ACM} Manag. Data}}
  \bibinfo{volume}{3}, \bibinfo{number}{2} (\bibinfo{year}{2025}),
  \bibinfo{pages}{98:1--98:26}.
\newblock
\href{https://doi.org/10.1145/3725235}{doi:\nolinkurl{10.1145/3725235}}


\bibitem[Khamis et~al\mbox{.}(2017)]%
        {DBLP:conf/pods/Khamis0S17}
\bibfield{author}{\bibinfo{person}{Mahmoud~Abo Khamis},
  \bibinfo{person}{Hung~Q. Ngo}, {and} \bibinfo{person}{Dan Suciu}.}
  \bibinfo{year}{2017}\natexlab{}.
\newblock \showarticletitle{What Do Shannon-type Inequalities, Submodular
  Width, and Disjunctive Datalog Have to Do with One Another?}. In
  \bibinfo{booktitle}{\emph{Proceedings of the 36th {ACM} {SIGMOD-SIGACT-SIGAI}
  Symposium on Principles of Database Systems, {PODS} 2017}}.
  \bibinfo{publisher}{{ACM}}, \bibinfo{pages}{429--444}.
\newblock
\href{https://doi.org/10.1145/3034786.3056105}{doi:\nolinkurl{10.1145/3034786.3056105}}


\bibitem[Khamis et~al\mbox{.}(2025b)]%
        {panda25}
\bibfield{author}{\bibinfo{person}{Mahmoud~Abo Khamis},
  \bibinfo{person}{Hung~Q. Ngo}, {and} \bibinfo{person}{Dan Suciu}.}
  \bibinfo{year}{2025}\natexlab{b}.
\newblock \showarticletitle{{PANDA:} Query Evaluation in Submodular Width}.
\newblock \bibinfo{journal}{\emph{TheoretiCS}}  \bibinfo{volume}{4}, Article
  \bibinfo{articleno}{12} (\bibinfo{year}{2025}).
\newblock
\href{https://doi.org/10.46298/THEORETICS.25.12}{doi:\nolinkurl{10.46298/THEORETICS.25.12}}


\bibitem[Khamis et~al\mbox{.}(2025c)]%
        {pandaexpress}
\bibfield{author}{\bibinfo{person}{Mahmoud~Abo Khamis},
  \bibinfo{person}{Hung~Q. Ngo}, {and} \bibinfo{person}{Dan Suciu}.}
  \bibinfo{year}{2025}\natexlab{c}.
\newblock \showarticletitle{PANDAExpress: a Simpler and Faster {PANDA}
  Algorithm}.
\newblock \bibinfo{journal}{\emph{CoRR}}  \bibinfo{volume}{abs/2512.10217}
  (\bibinfo{year}{2025}).
\newblock
\href{https://doi.org/10.48550/ARXIV.2512.10217}{doi:\nolinkurl{10.48550/ARXIV.2512.10217}}
\showeprint[arXiv]{2512.10217}


\bibitem[Lanzinger(2022)]%
        {LanzingerPODS22}
\bibfield{author}{\bibinfo{person}{Matthias Lanzinger}.}
  \bibinfo{year}{2022}\natexlab{}.
\newblock \showarticletitle{The Complexity of Conjunctive Queries with Degree
  2}. In \bibinfo{booktitle}{\emph{Proceedings of the 41st {ACM}
  {SIGMOD-SIGACT-SIGAI} Symposium on Principles of Database Systems, {PODS}
  2022}}. \bibinfo{publisher}{{ACM}}, \bibinfo{pages}{91--102}.
\newblock
\href{https://doi.org/10.1145/3517804.3524152}{doi:\nolinkurl{10.1145/3517804.3524152}}


\bibitem[Marx(2011)]%
        {adw}
\bibfield{author}{\bibinfo{person}{D{\'{a}}niel Marx}.}
  \bibinfo{year}{2011}\natexlab{}.
\newblock \showarticletitle{Tractable Structures for Constraint Satisfaction
  with Truth Tables}.
\newblock \bibinfo{journal}{\emph{Theory Comput. Syst.}} \bibinfo{volume}{48},
  \bibinfo{number}{3} (\bibinfo{year}{2011}), \bibinfo{pages}{444--464}.
\newblock
\href{https://doi.org/10.1007/S00224-009-9248-9}{doi:\nolinkurl{10.1007/S00224-009-9248-9}}


\bibitem[Marx(2013)]%
        {Marx13}
\bibfield{author}{\bibinfo{person}{D{\'a}niel Marx}.}
  \bibinfo{year}{2013}\natexlab{}.
\newblock \showarticletitle{Tractable Hypergraph Properties for Constraint
  Satisfaction and Conjunctive Queries}.
\newblock \bibinfo{journal}{\emph{J. {ACM}}} \bibinfo{volume}{60},
  \bibinfo{number}{6} (\bibinfo{year}{2013}), \bibinfo{pages}{42:1--42:51}.
\newblock
\href{https://doi.org/10.1145/2535926}{doi:\nolinkurl{10.1145/2535926}}


\bibitem[Menger(1927)]%
        {menger1927allgemeinen}
\bibfield{author}{\bibinfo{person}{Karl Menger}.}
  \bibinfo{year}{1927}\natexlab{}.
\newblock \showarticletitle{Zur allgemeinen kurventheorie}.
\newblock \bibinfo{journal}{\emph{Fundamenta mathematicae}}
  \bibinfo{volume}{10}, \bibinfo{number}{1} (\bibinfo{year}{1927}),
  \bibinfo{pages}{96--115}.
\newblock


\bibitem[Oum and Seymour(2006)]%
        {OumSeymour06}
\bibfield{author}{\bibinfo{person}{Sang{-}il Oum} {and}
  \bibinfo{person}{Paul~D. Seymour}.} \bibinfo{year}{2006}\natexlab{}.
\newblock \showarticletitle{Approximating clique-width and branch-width}.
\newblock \bibinfo{journal}{\emph{J. Comb. Theory {B}}} \bibinfo{volume}{96},
  \bibinfo{number}{4} (\bibinfo{year}{2006}), \bibinfo{pages}{514--528}.
\newblock
\href{https://doi.org/10.1016/J.JCTB.2005.10.006}{doi:\nolinkurl{10.1016/J.JCTB.2005.10.006}}


\bibitem[Reed(1997)]%
        {Reed97}
\bibfield{author}{\bibinfo{person}{B.~A. Reed}.}
  \bibinfo{year}{1997}\natexlab{}.
\newblock \showarticletitle{Tree Width and Tangles: A New Connectivity Measure
  and Some Applications}.
\newblock In \bibinfo{booktitle}{\emph{Surveys in Combinatorics, 1997}},
  \bibfield{editor}{\bibinfo{person}{R.~A. Bailey}} (Ed.).
  \bibinfo{series}{London Mathematical Society Lecture Note Series},
  Vol.~\bibinfo{volume}{241}. \bibinfo{publisher}{Cambridge University Press},
  \bibinfo{pages}{87--162}.
\newblock
\href{https://doi.org/10.1017/CBO9780511662119.006}{doi:\nolinkurl{10.1017/CBO9780511662119.006}}


\bibitem[Rockafellar(1997)]%
        {rockafellar1997convex}
\bibfield{author}{\bibinfo{person}{R~Tyrrell Rockafellar}.}
  \bibinfo{year}{1997}\natexlab{}.
\newblock \bibinfo{booktitle}{\emph{Convex analysis}}.
  Vol.~\bibinfo{volume}{28}.
\newblock \bibinfo{publisher}{Princeton university press}.
\newblock


\bibitem[Zhuk(2020)]%
        {Zhuk20}
\bibfield{author}{\bibinfo{person}{Dmitriy Zhuk}.}
  \bibinfo{year}{2020}\natexlab{}.
\newblock \showarticletitle{A Proof of the {CSP} Dichotomy Conjecture}.
\newblock \bibinfo{journal}{\emph{J. {ACM}}} \bibinfo{volume}{67},
  \bibinfo{number}{5} (\bibinfo{year}{2020}), \bibinfo{pages}{30:1--30:78}.
\newblock
\href{https://doi.org/10.1145/3402029}{doi:\nolinkurl{10.1145/3402029}}


\end{thebibliography}

 \appendix

\section{A Direct Lifted Submodular Width Lower Bound}
\label{app:lift.ex}

We demonstrate a simple mechanism for constructing symmetric edge-set profiles with large cut value and small submodular gauge.

\begin{lemma}[Orthogonal Partitions Lemma]
\label{thm:orthogonal-partitions}
Let $H$ be a hypergraph, and let
\[
\mathcal A=\{A_i \mid i\in I\},
\qquad
\mathcal B=\{B_j \mid j\in J\}
\]
be two families of edges of $H$, each of which is a partition of $V(H)$.
Let $\mu$ be a probability measure on $V(H)$ such that
$\mu(A_i\cap B_j)=\mu(A_i)\mu(B_j)$
 for all $i\in I, j\in J$.
Set $\theta:=\max_{e\in E(H)}\mu(e)$,
and assume that $\mu(A_i)\le \tfrac12$ for all $i\in I$.

Then
\(
\subw(H)\ge \frac{1}{3\theta}.
\)
\end{lemma}

\begin{proof}
Set $b(X):=\mu(X)/\theta$ for $X\subseteq V(H)$. Since $b$ is modular and
$b(X)\le 1$ whenever $X\subseteq e\in E(H)$, we have
$b\in \pMod{\one}(H)\subseteq \pSub{\one}(H)$. Let $q:=\lambda_b$, so
$q(F)=\mu(\partial_H(F))/\theta$ for every $F\subseteq E(H)$.

Define a probability measure $\rho$ on $E(H)$ by
\[
\rho(e):=
\begin{cases}
\mu(e), & e\in \mathcal A,\\
0, & e\notin \mathcal A.
\end{cases}
\]
Since $\mathcal A$ is a partition of $V(H)$, $\rho$ is a probability measure on
$E(H)$, and $\rho(e)\le \tfrac12$ for every $e\in E(H)$.

Fix $F\subseteq E(H)$, and set
\[
A_F:=\bigcup (\mathcal A\cap F),
\qquad
B_F:=\bigcup (\mathcal B\cap F).
\]
If $x\in A_F\triangle B_F$ ($\triangle$ being the symmetric set difference), then one of the two partition-edges containing $x$
lies in $F$ and the other in $E(H)\setminus F$. Thus
$x\in \partial_H(F)$, and hence
\(
\mu(\partial_H(F))\ge \mu(A_F\triangle B_F).
\)

Because $\mathcal A$ and $\mathcal B$ are partitions and
$\mu(A_i\cap B_j)=\mu(A_i)\mu(B_j)$ for all $i,j$, the events
$x\in A_F$ and $x\in B_F$ are independent under $\mu$. Moreover,
$\mu(A_F)=\rho(F)$ and we have
\[
\mu(A_F\triangle B_F)
=
\mu(A_F)+\mu(B_F)-2\mu(A_F)\mu(B_F)
=
\rho(F)+\mu(B_F)(1-2\rho(F)).
\]
Therefore, whenever $\rho(F)\le \tfrac12$, we also have
$\mu(\partial_H(F))\ge \rho(F)$.
By symmetry of $q$ it follows
that
\[
q(F)\ge \frac{1}{\theta}\min\{\rho(F),\rho(E(H)\setminus F)\}
\qquad \forall F\subseteq E(H).
\]

Applying \Cref{lem:general-balanced-profile} with
$\phi(x)=x/\theta$ yields $\disp_H(q)\ge 1/(3\theta)$.
Since $q=\lambda_b$ with $b\in\pSub{\one}(H)$, we get
$\subwlift(H)\ge \disp_H(q)\ge 1/(3\theta)$.
Finally, by \Cref{cor:lifted-bw} also $\subw(H)\ge \subwlift(H)$.
\end{proof}

A simple example is the row/column hypergraph of an $m\times n$ grid. Its
vertex set is $[m]\times [n]$, and its hyperedges are the rows
$R_i=\{(i,j)\mid j\in [n]\}$  for $i\in [m]$, 
and the columns
$C_j=\{(i,j)\mid i\in [m]\}$ for $j\in [n]$.

Taking $\mathcal A=\{R_i\mid i\in [m]\}$ and $\mathcal B=\{C_j\mid j\in [n]\}$,
we obtain two partitions of the vertex set. If $\mu$ is the uniform measure on
$[m]\times [n]$, then
\[
\mu(R_i)=\frac1m,\qquad \mu(C_j)=\frac1n,\qquad
\mu(R_i\cap C_j)=\frac1{mn}=\mu(R_i)\mu(C_j),
\]
so the two partitions are orthogonal in the sense of
\Cref{thm:orthogonal-partitions}. Moreover,
\[
\max_{e\in E(H)}\mu(e)=\max\Bigl\{\frac1m,\frac1n\Bigr\}
=\frac1{\min\{m,n\}},
\]
and if $m,n\ge 2$, then in particular each row has measure at most $1/2$.
Thus \Cref{thm:orthogonal-partitions} applies and yields
\[
\subw(H)\ge \frac{\min\{m,n\}}{3}.
\]
Somewhat surprisingly, \Cref{thm:orthogonal-partitions} also tells us that the same lower bound holds if we add arbitrary additional hyperedges on the same
vertex set, as long as they all have rank at most $\max\{m,n\}$.

\section{Additional Details for Section~\ref{sec:applications}}
\label{app:applications}

We give individual proofs for the bounds on $\cgamp$ in terms of excess and private edge intersections. While both cases are covered by the more general argument in~\Cref{sec:kappa} we believe the more specific arguments here can be instructive.

Both proofs apply the following principle.
\begin{lemma}
\label{lem:witness-set-realisation}
Let $H$ be a hypergraph, let $M:=P(H^*)$, and let $\alpha\in \rpos^{E(M)}$.
For each edge $xy\in E(M)$, let $A_{xy}\subseteq x\cap y$ be nonempty.
Define
\[
L_A(\alpha):=
\max_{e\in E(H)}
\sum_{\substack{xy\in E(M)\\A_{xy}\cap e\neq \emptyset}}
\alpha_{xy}.
\]
Then $\gauge_H^\Sub(q_{M,\alpha})\le L_A(\alpha)$.
\end{lemma}

\begin{proof}
Our goal is to show that for every $\alpha$, there is a $b\in\Sub(H)$ with $q_{M,\alpha} \le \lambda_b$ and $\lambda_b \in L_A(\alpha)\cdot K^\Sub_\one$.

To this end, define
\[
b(X):=
\sum_{xy\in E(M)} \alpha_{xy}\,\mathbf 1[A_{xy}\cap X\neq \emptyset]
\qquad \forall X\subseteq V(H)
\]
where $\one[\varphi]$ is $1$ if $\varphi$ is true, and $0$ otherwise. This is a coverage function, hence $b\in \Sub(H)$.

Fix $F\subseteq E(H)$. If $xy\in \cut{M}(F)$, then $x$ and $y$ lie on opposite
sides of the cut, so every vertex of $x\cap y$, and therefore every vertex of
$A_{xy}$, belongs to $\partial_H(F)$. Hence
\[
q_{M,\alpha}(F)
=
\sum_{xy\in \cut{M}(F)}\alpha_{xy}
\le
b(\partial_H(F))
=
\lambda_b(F).
\]

What is left to show is that $\lambda_b \in L_A(\alpha)\cdot K^\Sub_\one$. Towards this claim, observe that for every $e\in E(H)$,
\[
b(e)
=
\sum_{xy\in E(M)} \alpha_{xy}\,\mathbf 1[A_{xy}\cap e\neq \emptyset]
=
\sum_{\substack{xy\in E(M)\\A_{xy}\cap e\neq \emptyset}}\alpha_{xy}
\le
L_A(\alpha).
\]
If $L_A(\alpha)=0$, then $q_{M,\alpha}=0$ and there is nothing to prove.
Otherwise, $\tilde b:=L_A(\alpha)^{-1}b$ lies in $\pSub{\one}(H)$, and
$q_{M,\alpha}\le \lambda_b = L_A(\alpha)\lambda_{\tilde b}$.
Thus $q_{M,\alpha}\in L_A(\alpha)\,K^\Sub_{\one}(H)$, which means
$\gauge_H^\Sub(q_{M,\alpha})\le L_A(\alpha)$.
\end{proof}

\begin{proposition}
\label{prop:exclusive-witness-realisation}
Assume that $H$ has the private intersections property. Let
$M:=P(H^*)$, and let $\alpha\in \rpos^{E(M)}$. Then
\(
\gauge_H^\Sub(q_{M,\alpha})
\le
\max_{e\in E(H)}\alpha(\Inc_M(e)).
\)
\end{proposition}

\begin{proof}
For each edge $xy\in E(M)$, fix a witness vertex
$v_{xy}\in x\cap y$ with $\deg_H(v_{xy})=2$, and set $A_{xy}:=\{v_{xy}\}$.
Set $\tau:=\max_{e\in E(H)}\alpha(\Inc_M(e))$. Our goal will be to show that $L_A(\alpha)\le \tau$ ---  in the sense of \Cref{lem:witness-set-realisation} --- for these $A_{xy}$.

Fix $e\in E(H)$. We want to understand which terms can contribute to the load at
$e$, that is, for which edges $xy\in E(M)$ we have $A_{xy}\cap e\neq\emptyset$.
Since $A_{xy}=\{v_{xy}\}$, this simply means that the witness vertex $v_{xy}$
lies in $e$.

Now $v_{xy}$ was chosen from $x\cap y$, so it already lies in the two hyperedges
$x$ and $y$. Moreover, by assumption $\deg_H(v_{xy})=2$, so these are the only
two hyperedges of $H$ containing $v_{xy}$. Therefore, if $v_{xy}\in e$, then
necessarily $e=x$ or $e=y$.

In other words, a witness vertex can only be seen by the two hyperedges that
gave rise to it. Hence every term counted in
\[
\sum_{\substack{xy\in E(M)\\A_{xy}\cap e\neq\emptyset}}\alpha_{xy}
\]
comes from an edge $xy$ of the line graph $M$ that is incident with $e$. It
follows that
\[
\sum_{\substack{xy\in E(M)\\A_{xy}\cap e\neq\emptyset}}\alpha_{xy}
\le
\alpha(\Inc_M(e))
\le
\tau.
\]
Therefore $L_A(\alpha)\le \tau$, and the claim follows from
\Cref{lem:witness-set-realisation}.
\end{proof}

\begin{proposition}
\label{bounded-excess-realisation}
Let $H$ be a hypergraph, let $M:=P(H^*)$, and let $\alpha\in \rpos^{E(M)}$.
Then
\[
\gauge_H^\Sub(q_{M,\alpha})
\le
(\eex(H)+1)\max_{e\in E(H)}\alpha(\Inc_M(e)).
\]
\end{proposition}

\begin{proof}
Our goal will be to apply \Cref{lem:witness-set-realisation} with
$A_{xy}:=x\cap y$ for all edges $xy\in E(M)$.
Set $\tau:=\max_{e\in E(H)}\alpha(\Inc_M(e))$.
We will show that $L_A(\alpha)\le (\eex(H)+1)\tau$ in the sense of \Cref{lem:witness-set-realisation}.

Fix an $e\in E(H)$, write
\[
L_e:=
\sum_{\substack{xy\in E(M)\\(x\cap y)\cap e\neq \emptyset}}\alpha_{xy}.
\]
Observe that the edges of the line graph incident with $e$, i.e., $\Inc_M(e)$ contribute at most
$\alpha(\Inc_M(e))\le \tau$.

For the remaining edges of the line graph, namely any $xy \in E(M) \setminus \Inc_M(e)$,  if $(x\cap y)\cap e\neq \emptyset$, then
there exists a vertex $v\in V(H)$ such that $v\in e\cap x\cap y$. Define
$S_v:=\{h\in E(H)\setminus\{e\}\mid v\in h\}$.
Then $x,y\in S_v$, so $xy\in E(M[S_v])$. Hence
\[
L_e
\le
\tau+\sum_{v\in e}\sum_{xy\in E(M[S_v])}\alpha_{xy} \tag{$\ast$}.
\]

If $\deg_H(v)\le 2$, then $|S_v|\le 1$, so $E(M[S_v])=\emptyset$. That is, there is no $xy \in E(M)\setminus \Inc_M(e)$ with $v\in e \cap x \cap y$. Alternatively, since $e,x,y$ are distinct, $v$ in their intersection would imply that $v$ has degree 3.

Thus, only vertices with degree at least $3$ contribute to the sum in $(\ast)$. Suppose that $\deg_H(v)\ge 3$, then by standard double counting
\[
2\sum_{xy\in E(M[S_v])}\alpha_{xy}
\le
\sum_{h\in S_v}\alpha(\Inc_M(h))
\le
|S_v|\tau
=
(\deg_H(v)-1)\tau.
\]
Dividing by $2$, we obtain
\[
\sum_{xy\in E(M[S_v])}\alpha_{xy}
\le \tau 
\frac{\deg_H(v)-1}{2}\le 
\tau \max\{\deg_H(v)-2,0\}.
\]
Then combining with $(\ast)$ from above we get the final inequality,
\[
L_e
\le
\tau+\tau\sum_{v\in e}\max\{\deg_H(v)-2,0\}
\le
(\eex(H)+1)\tau.
\]
Taking the maximum over $e$ gives $L_A(\alpha)\le (\eex(H)+1)\tau$, and the claim follows from \Cref{lem:witness-set-realisation}.
\end{proof}

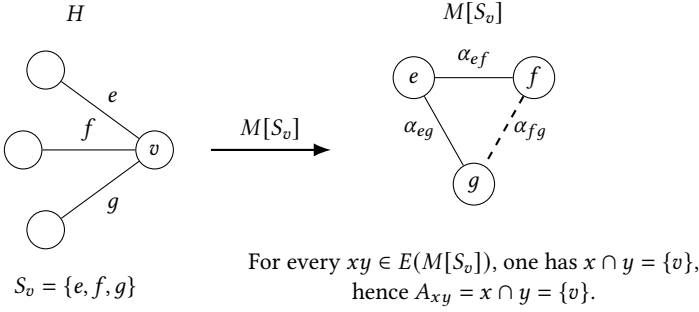
\begin{figure}[t]
\centering
\begin{tikzpicture}[
  >=Latex,
  vertex/.style={circle,draw,fill=white,minimum size=5mm,inner sep=0pt},
  edgev/.style={circle,draw,fill=white,minimum size=5.5mm,inner sep=0pt},
  every node/.style={font=\small}
]

\node[vertex,label=center:$v$] (v) at (0,0) {};
\node[vertex] (a) at (-1.45,1.05) {};
\node[vertex]  (b) at (-1.75,0) {};
\node[vertex] (c) at (-1.45,-1.05) {};

\draw (v) -- node[midway,above right=-1pt] {$e$} (a);
\draw (v) -- node[midway,above] {$f$} (b);
\draw (v) -- node[midway,below right=-1pt] {$g$} (c);

\node at (-1.05,1.8) {$H$};
\node[align=center] at (-1.05,-1.8) {$S_v=\{e,f,g\}$};

\draw[->,thick] (0.75,0) -- (2.35,0) node[midway,above] {$M[S_v]$};

\node[edgev] (E) at (3.45,0.95) {$e$};
\node[edgev] (F) at (5.05,0.95) {$f$};
\node[edgev] (G) at (4.25,-0.45) {$g$};

\draw (E) -- node[midway,above] {$\alpha_{ef}$} (F);
\draw (E) -- node[midway,left] {$\alpha_{eg}$} (G);
\draw[dashed,thick] (F) -- node[midway,right] {$\alpha_{fg}$} (G);

\node at (4.25,1.8) {$M[S_v]$};
\node[align=center] at (4.25,-1.7) {
For every $xy\in E(M[S_v])$, one has $x\cap y=\{v\}$,\\
hence $A_{xy}=x\cap y=\{v\}$.};

\end{tikzpicture}

\caption{
Illustration for the bounded-excess proof for a vertex $v$ of degree $3$. In the bound of $L_e$, $\alpha_{ef}+\alpha_{eg}$ is bounded by $\alpha(\Inc_M(e))$, whereas the dashed edge $fg$ also has $A_{fg}\cap e \neq \emptyset$ and thus also additionally affects the constructed submodular witness in~\Cref{lem:witness-set-realisation}.
}
\label{fig:bounded-excess-local}
\end{figure}

\section{A Natural Family with Logarithmic Edge Excess}
\label{app:excess}

We present a natural family of hypergraphs with edge excess $O(\log\ghw)$, illustrating that even superconstant edge excess suffices to obtain a strong lower bound on $\subw$ in terms of $\ghw$.

For $r\ge 2$, let $Q_r$ be the $r$-dimensional cube on vertex set
$(\mathbb F_2)^r$. We define a hypergraph $H_r$ as follows.

\begin{enumerate}
\item The vertices of $H_r$ are the edges of $Q_r$.
\item For each cube vertex $x\in (\mathbb F_2)^r$, we add the star
\[
S_x:=\{f\in E(Q_r) \mid x\in f\}.
\]
\item For each $i\in\{2,\dots,r\}$ and each $x\in(\mathbb F_2)^r$ with
$x_1=x_i=0$, we add the square
\[
C_{x,i}:=E\bigl(Q_r[\{x,x+e_1,x+e_i,x+e_1+e_i\}]\bigr).
\]
\end{enumerate}

So $H_r$ has one hyperedge for every star of the cube, and one hyperedge for
every $2$-dimensional face in the directions $(e_1,e_i)$.

We use the following terminology of Adler, Gottlob, and Grohe~\cite{DBLP:journals/ejc/AdlerGG07}.
For a hypergraph $H$, a connected set $C\subseteq V(H)$, and a set
$M\subseteq E(H)$, we say that $C$ is \emph{$M$-big} if
\[
\bigl|\{e\in M \mid e\cap C\neq\emptyset\}\bigr|>\frac{|M|}{2}.
\]
For $k\ge 0$, a set $M\subseteq E(H)$ is \emph{$k$-hyperlinked} if for every
$S\subseteq E(H)$ with $|S|<k$, the hypergraph
\[
H\setminus \bigcup S \;:=\; H\bigl[V(H)\setminus \bigcup_{e\in S} e\bigr]
\]
has an $M$-big connected component.
The \emph{hyperlinkedness} $\mathsf{hlink}(H)$ is the largest $k$ such that $H$
contains a $k$-hyperlinked set.
It was shown in~\cite{DBLP:journals/ejc/AdlerGG07} that always $\mathsf{hlink}(H)\le \ghw(H)$.

\begin{proposition}
For every $r\ge 4$,
\[
\eex(H_r)\le 2r
\qquad\text{and}\qquad
\ghw(H_r) \in \Omega\left(\frac{2^r}{r}\right).
\]
Hence $\eex(H_r)=O(\log \ghw(H_r))$,
and therefore
\[
\subw(H_r)\ge C\,\frac{\ghw(H_r)}{\log^2 \ghw(H_r)}
\]
for all sufficiently large $r$.
\end{proposition}

\begin{proof}
We first compute the edge excess. A cube edge in direction $e_1$ lies in its two
endpoint stars and in one square of type $(e_1,e_i)$ for each $i=2,\dots,r$, so
its degree in $H_r$ is $r+1$. A cube edge in direction $e_i$ with $i\ge 2$ lies
in its two endpoint stars and in exactly one square, so its degree is $3$.
Therefore, for every $x\in (\mathbb F_2)^r$,
\[
\sum_{f\in S_x}\max\{\deg_{H_r}(f)-2,0\}
=
(r+1-2)+(r-1)(3-2)=2r-2,
\]
and for every square $C_{x,i}$,
\[
\sum_{f\in C_{x,i}}\max\{\deg_{H_r}(f)-2,0\}
=
2(r+1-2)+2(3-2)=2r.
\]
Hence $\eex(H_r)\le 2r$.

We move on to the lower bound on $\ghw(H_r)$.
Let $M:=\{S_x \mid x\in (\mathbb F_2)^r\}$,(so $|M|=2^r$),
and let
$k:=\left\lfloor 2^{r-2}/r\right\rfloor$.
We claim that $M$ is $k$-hyperlinked. Since $\mathsf{hlink}(H_r)\le \ghw(H_r)$,
this will imply $\ghw(H_r)\ge k$.

Fix $S\subseteq E(H_r)$ with $|S|<k$.
Recall that the vertices of $H_r$ are the cube edges of $Q_r$.
Thus $\bigcup S\subseteq V(H_r)=E(Q_r)$ is exactly the set of cube edges
that lie in one of the selected hyperedges of $S$.
Since every hyperedge of $H_r$ has size at most $r$ (stars have size $r$,
squares have size $4$), we have
\[
\left|\bigcup S \right|
\le \sum_{h\in S}|h|
\le r|S|
< kr
\le 2^{r-2}.
\]
Equivalently, fewer than $2^{r-2}$ cube edges are deleted.

Let $G$ be the spanning subgraph of $Q_r$ whose edge set is
\(
E(G)=E(Q_r)\setminus \bigcup S.
\)

\medskip
\noindent\textbf{Claim 1.}
$G$ has a connected component $U$ with $|U|>2^{r-1}$.

\smallskip
\noindent\emph{Proof of Claim 1.}
Suppose not. Then every connected component of $G$ has size at most $2^{r-1}$.
Choose a union $X$ of connected components of $G$ with
$2^{r-2}\le |X|\le 2^{r-1}$: either one component already has size in this
range, or all components have size less than $2^{r-2}$, in which case such an
$X$ can be built greedily. Because $X$ is a union of connected components, every
cube edge in $\delta_{Q_r}(X)$ was deleted. Hence
$|\delta_{Q_r}(X)|<2^{r-2}$. On the other hand, the edge-isoperimetric
inequality for the cube gives
\[
|\delta_{Q_r}(X)|
\ge
|X|\log_2\!\left(\frac{2^r}{|X|}\right)
\ge
2^{r-2},
\]
a contradiction.\hfill $\blacktriangle$

\medskip
\noindent\textbf{Claim 2.}
$H_r\setminus \bigcup S$ has a connected component $C$
that meets every star $S_x$ with $x\in U$.

\smallskip
\noindent\emph{Proof of Claim 2.}
Let $W$ be the set of cube edges of $G$ with both endpoints in $U$.
Because $U$ is a connected component of $G$ and $|U|>2^{r-1}$, the graph
$G[U]$ is nontrivial, so every $x\in U$ is incident with some edge of $W$.
In particular, $S_x\notin S$ for every $x\in U$.

We claim that $W$ lies in a single connected component of
$H_r\setminus \bigcup S$.
Indeed, let $e,f\in W$.
Since $G[U]$ is connected, there is a sequence of cube edges
\[
e=e_0,e_1,\dots,e_t=f
\]
in $W$ such that consecutive edges $e_{i-1},e_i$ share a cube vertex
$x_i\in U$.
Then both $e_{i-1}$ and $e_i$ belong to the star $S_{x_i}$, and since
$S_{x_i}\notin S$, the hyperedge $S_{x_i}$ survives in
$H_r\setminus \bigcup S$.
Hence $e_{i-1}$ and $e_i$ are adjacent in the primal graph of
$H_r\setminus \bigcup S$.
Therefore all edges in $W$ lie in one connected component $C$ of
$H_r\setminus \bigcup S$.

Finally, for each $x\in U$, some edge of $W$ is incident with $x$, so that
edge belongs to $C\cap S_x$.
Thus $C$ meets every star $S_x$ with $x\in U$.
\hfill $\blacktriangle$

\medskip
By Claim~1, we have $|U|>2^{r-1}=|M|/2$.
By Claim~2, there is a connected component $C$ of $H_r\setminus \bigcup S$
meeting every star $S_x$ with $x\in U$.
Since the stars $S_x$ are pairwise distinct, $C$ meets more than half of the
hyperedges in $M$.
Hence $C$ is $M$-big.

We have proved that for every $S\subseteq E(H_r)$ with $|S|<k$, the
hypergraph $H_r\setminus \bigcup S$ has an $M$-big connected component.
Therefore $M$ is $k$-hyperlinked, so $\mathsf{hlink}(H_r)\ge k$ and therefore
\[
\ghw(H_r)\ge k=\left\lfloor \frac{2^{r-2}}{r}\right\rfloor
\in \Omega\!\left(\frac{2^r}{r}\right).
\]

Therefore
$\eex(H_r) \le 2r=O(\log \ghw(H_r))$ and applying \Cref{bounded-excess-ratio} and \Cref{lem:load.to.width}  yields
\[
\subw(H_r)\ge C\,\frac{\ghw(H_r)}{\log^2 \ghw(H_r)}
\]
for some universal constant $C>0$ and all sufficiently large $r$.
\end{proof}

Note that the rank of $H_r$ is $r$, since each star $S_x$ has size $r$. That is, the family of hypergraphs studied here has unbounded rank and edge excess.

\end{document}